\documentclass[12pt]{article}
\usepackage{geometry}

\usepackage{subcaption} 
\usepackage[linesnumbered,ruled,vlined]{algorithm2e}
\usepackage{graphicx,xcolor}
\usepackage{amsmath,amsthm,amssymb,bm}
\usepackage{natbib}
\usepackage{enumitem}
\usepackage{multirow}
\usepackage{placeins}
\usepackage[normalem]{ulem}
\newcommand{\suit}[1]{\left( #1\right)}
\newcommand{\msuit}[1]{\left[ #1\right]}
\newcommand{\set}[1]{\left\{ #1\right\}}
\newcommand{\abs}[1]{\left| #1\right|}

\newcommand{\mI}{\mathcal{I}}

\newcommand{\mS}{\mathcal{S}}

\newcommand{\mM}{\mathcal{M}}
\newcommand{\bE}{\mathbb{E}}
\newcommand{\Var}{\textrm{Var}}

\newcommand{\bolSigma}{\boldsymbol{\Sigma}}

\newcommand{\var}{\mathrm{Var}}
\newcommand{\sfmax}{ s_{n,\max}^{(4)} }
\newcommand{\ssmax}{ s_{n,\max}^{(2)} }
\newcommand{\ssmin}{ s_{n,\min}^{(2)} }

\theoremstyle{plain}
\newtheorem{lemma}{Lemma}

\newtheorem{theorem}{Theorem}
\newtheorem{corollary}{Corrollary}

\title{High Dimensional Mean Test for Shrinking Random Variables with Applications to Backtesting}
\author{Liujun Chen \and Chen Zhou}

\begin{document}

\maketitle

\begin{abstract}
We propose a high dimensional mean test framework for shrinking random variables, where the underlying random variables shrink to zero as the sample size increases. By pooling observations across overlapping subsets of dimensions, we estimate subsets means and test whether the maximum absolute mean deviates from zero. This approach overcomes cancellations that occur in simple averaging and remains valid even when marginal asymptotic normality fails. We establish  theoretical properties of the test statistic and develop a multiplier bootstrap procedure to approximate its distribution. The method provides a flexible and powerful tool for the  validation and comparative backtesting of value-at-risk. Simulations show superior performance in high-dimensional settings, and a real-data application demonstrates its practical effectiveness in backtesting.
	\end{abstract}

{\it Keywords}: data pooling,  multiplier bootstrap, value-at-risk

\section{Introduction}
In both financial and climate risk management, the prediction accuracy of tail risks is a fundamental concern.
A key instrument for validating the reliability of tail risk forecasts is backtesting \citep{Chris2003elements, Mcneil2005quant}.
 That is, checking how well the predicted tail risk measures, such as  extreme quantiles, perform when confronted with new data. When  dealing with  large financial losses and  extreme climate events, a major challenge in backtesting is the scarcity of tail observations.

For example, in the contexts of Basel regulation for trading books \citep{BCBS016standard}, backtesting often involves assessing whether the predicted 99\% quantile of losses is accurate. Typically, using new data with a sample size $n=250$, reflecting 250 trading days in a year, accurately forecasting one-percent quantiles at each day corresponds to roughly 2.5 exceedances per year.  However, the effectiveness of such tests is limited, as the expected number of exceedances is too small. With only about two to three violations per year, the asymptotic approximations employed in standard backtesting procedures become unreliable, leading to unstable inference.

In essence, violation indicators are regarded as independent and identically distributed (i.i.d.) Bernoulli random variables if the forecasts of the conditional quantiles are accurate, with the violation probability being their common mean. Testing an extremely low violation probability with few tail observations is equivalent to conducting a mean zero test where the underlying Bernoulli random variables  shrink to zero in probability.

To overcome the limitation of insufficient data, one natural approach is to leverage data from multiple sources. In the financial contexts, one could consider  multiple assets, while in the contexts of climate risk, data from multiple gauging stations can be used. Combining data from many sources results in a high-dimensional dataset, which contains $n$ observations with dimension $p$. A good risk forecasting method should perform well across all dimensions.  After combining data, tail observations in such a high-dimensional dataset are not scarce.

Following this idea, 
the simplest approach is 
to conduct a pooling mean test using high-dimensional data where the underlying random variables shrink to zero. This testing problem differs from the typical mean zero test in high-dimensional statistics where the key building block is the asymptotic normality of the mean estimator at each dimension. Due to the shrinkage of the underlying random variables, asymptotic normality of the mean estimator at each dimension fails. By contrast, pooling data from all $p$-dimensions yields a reliable estimate for the average mean across these dimensions. Subsequently, one can  test whether such an average mean is significantly different from zero.
However, this simple approach bears an obvious shortcoming:  the null of having an average  mean zero is not equivalent to our intended null of having mean zeros at all dimensions. In other words, marginal means with different signs may cancel out each other in the average.

To  achieve the test of mean zeros at all dimensions, we propose a subsets-based approach. 
Consider $d$ subsets of the $p$ dimensions, each containing  $q$ dimensions. Those subsets can overlap. 
Firstly, we estimate the average mean parameter using pooled data for each subset, resulting in $d$ estimates. Next, we test whether the maximum absolute value among the $d$ estimates differs significantly from zero. This is achieved by establishing the asymptotic behavior of the test statistic and developing a multiplier bootstrap procedure to approximate the limit distribution.

 This subsets-based pooling approach introduces both advantages and new  challenges. We show that, with a suitable choice of the subsets, having  
  average mean zeros for all subsets is equivalent to having mean zeros at all dimensions. A methodological challenge is in the design of  the subsets, for which we provide practical guidance. Another challenge is in developing the theoretical guarantee for the test. Note that, the data across the $p$ dimensions are  not independent.  
  We provide the conditions under which such a testing method is theoretically guaranteed. 

Our contribution is threefold. First, we extend the literature on mean testing in high-dimensional settings by developing a method that remains valid even when underlying variables shrink to zero and when the asymptotic normality of the mean estimator at each marginal dimension fails. 
Over the past three decades, high-dimensional mean testing has received considerable attention \citep{bai1996effect, chen2010two, tony2014two, chang2017simulation}, and we refer interested readers to \cite{huang2022overview} for a recent review. However, most existing approaches typically rely on the assumption that the sample mean at each dimension satisfies a marginal central limit approximation \citep{chernozhukov2017central, chernozhukov2022improved}. Our method relaxes this assumption, thereby allowing for broader applicability in scenarios where marginal asymptotic normality does not hold.

Second, our approach is related to multi-task and transfer learning, where information is borrowed across related tasks or domains to improve estimation. 
Multi-task learning aims to improve generalization by leveraging shared information across related tasks \citep{Wang2018cross, Duan2023adaptive, Sui2025multi}, while transfer learning focuses on transferring knowledge from a source domain to a target domain \citep{Li2022transfer, Cai2024transfer, Li2024estimation}. 
  Our subsets-based pooling approach adopts a similar idea: by combining data across multiple dimensions within each subset, we effectively borrow strength across tasks, leading to more 
  precise estimation of the common mean of the shrinking random  variables.

Finally, our work contributes to the literature on backtesting methodologies. Various approaches have been developed to evaluate the validity of risk forecasting models \citep{Chris1998evaluate, mcneil2000estimation, christoffersen2004backtesting}, and recent studies have introduced comparative backtesting frameworks for statistically assessing competing forecasts \citep{fissler2015expected, nolde2017elicitability}. Our proposed method offers  a more powerful and flexible tool  for model validation  regarding extreme risk measures, and owing to its foundation in mean testing, can be naturally extended to comparative backtesting across different risk forecasting techniques.

The rest of the paper is organized as follows. In Section \ref{sec:method}, we consider a generalized mean test problem when the random variables shrink to zero. In Section \ref{sec:backtest}, we discuss the application of these methods to the VaR backtesting problem. We show the finite sample performance of the proposed methods via a simulation study in Section \ref{sec:simulation}. A real data application is given in Section \ref{sec:application}. All the technical proofs are gathered in the Supplementary Material.  For two sequences $a_n, b_n$,  $a_n \gtrsim b_n$ means that, $b_n/a_n =O(1)$,  
$a_n \asymp b_n$ means that, $a_n/b_n = O(1)$ and $b_n/a_n = O(1)$ as $n\to\infty$.

\section{Methodology}\label{sec:method}

For each sample size $n$, consider independent random vectors $(X^{(n)}_{i,1},   \dots, X^{(n)}_{i,p})$, $i=1,\dots,n$. 
For each dimension $j=1,\dots, p$, assume that the independent  random  variables $X^{(n)}_{1,j}, \dots, X^{(n)}_{n,j}$ share a common mean, i.e.,   
$$
\bE X^{(n)}_{1,j} = \bE X^{(n)}_{2,j} = \dots  = \bE X^{(n)}_{n,j} = : \mu^{(n)}_j.
$$
Our primary interest lies in testing whether all these means are zero, namely,
$$
H_0: \mu^{(n)}_j  = 0, \quad   j=1,\dots,p.
$$
In this paper, we focus on the scenario when the underling random variables shrink to zero as sample size increases, i.e.,  $X^{(n)}_{i,j} = o_P(1)$ as $n\to\infty$. In this case,  
 the normalized sums 
$$
\frac{\sum_{i=1}^n (X^{(n)}_{i,j} - \mu^{(n)}_j) }{ \sqrt{\sum_{i=1}^n\bE (X^{(n)}_{i,j}-\mu^{(n)}_j)^2   } } 
$$
may not  converge to a normal distribution. Consequently, standard mean test procedures in the literature fail. 

Without loss of generality, in the following, we drop the superscript $n$ for notational simplicity. That is, we write $X_{i,j}$ and $\mu_j$ instead of $X^{(n)}_{i,j}$ and $\mu^{(n)}_j$, respectively.

\subsection{A naive test}\label{sec:naive:test}
To address the challenge posed by shrinking  random variables, a natural idea is to pool information across all dimensions as follows. Define the aggregated random variables
$$
\begin{aligned}
	Y_i =  \sum_{j=1}^p X_{i,j},  \quad i=1,\dots,n.
\end{aligned}
$$
Note that, $Y_1,\dots, Y_n$ are independent random variables with common means $\mu_Y:=\sum_{j=1}^p \mu_j$.   The null hypothesis, $H_0: \mu_j=0, j=1,\ldots,p$, implies a different but simpler null hypothesis $H'_0: \mu_Y=0$.  We test this simpler null by considering the following test statistic 
$$
T = \frac{\sum_{i=1}^n Y_i }{\sqrt{n \widehat{\sigma}} },
$$
where 
$$
\widehat{\sigma} = \frac{1}{n}\sum_{i=1}^n \suit{Y_i - \frac{1}{n}\sum_{i=1}^n Y_i  }^2.
$$

To establish the asymptotic theory of the test statistic, we  need the following conditions.

Firstly, we  regulate the dependency  across random variables in the random vector $(X_{i,1}, \dots, X_{i,p})$, $i=1,\dots,n$,  by a $\rho$-mixing condition.  The concept of $\rho$-mixing was first developed by \cite{kolmogorov1960strong}. Define 
$$
\rho_{i}(m) =\sup_{1\le k\le p-m, Z_1\in \mathcal{F}_{i,1,k}, Z_2\in \mathcal{F}_{i,k+m,p} } \abs{\frac{\text{Cov}(Z_1,Z_2)}{\suit{\bE Z_1^2 \bE Z_2^2}^{1/2}}},
$$
where 
$\mathcal{F}_{i,a,b}$ is the $\sigma$-field generated by the random variables $X_{i,a}, X_{i,a+1} \dots,X_{i,b}$.
We assume that the following inequality holds.

\begin{enumerate}[label=(A)]
	\item \label{assum:mixing}
 For some $c_0>0$, such that, 
	$$
	\rho_{i}(m) \le \exp\suit{-2c_0m},
	$$
	uniformly for all $1\le i \le n$.
\end{enumerate}

The $\rho$-mixing condition  depends on the ordering of variables  which often has a natural meaning in the contexts of time series analysis. By contrast, beyond time series, for other data types, such as microarray data and spatial data, there
is  no natural ordering of the variables.
Following \cite{jung2009pca}, \cite{Chang2021double} and \cite{liu2024nonparametric}, we can weaken the condition \ref{assum:mixing} by assuming the existence of a 
  permutation $\pi_p$: $\set{1,\dots,p}\to \set{1,\dots,p}$, such that the permuted vector $\suit{X_{i,\pi_p(1)},\dots, X_{i,\pi_p(p)}  } $ is $\rho$-mixing.   This assumption makes the results invariant under a permutation of the variables.  Without loss of generality, we continue with the slightly stronger notion as in  Condition \ref{assum:mixing}. 
Note that, when the random vectors form an $m$-dependent sequence, they satisfy the $\rho$-mixing condition in Assumption \ref{assum:mixing} \citep{bradley2005basic}.

Next, we assume the following lower bound on the variance of $Y$.
\begin{enumerate}[label=(B)]
	\item \label{assum:var:low:bound} There exists a constant $c>0$, such that, 
$$
 \textnormal{Var}\suit{Y_i} \ge c \sum_{j=1}^p  \var(X_{i,j}),
$$
uniformly for all $1\le i\le n$.
\end{enumerate}

Condition \ref{assum:var:low:bound} requires that, the variance of    $\sum_{j=1}^ p X_{i,j}$ is at least of the same order, as the sum of the individual variances of its components. This assumption rules out strongly negative multicollinearity among the variables.
 This condition fails when some variables are nearly perfectly negatively correlated, e.g., $X_{i,2k} = -X_{i,2k-1}$, $k=1,2,\ldots$.  In such cases, the covariance terms in the expansion of $\var(Y_i)$ cancel out most of the marginal variances, leading to a much smaller overall variance of 
$Y_i$. As long as such strongly negative multicolinearity is avoided, Condition \ref{assum:var:low:bound} is satisfied.

Thirdly, we impose some uniform moment constraints on the observations, focusing on the second and fourth moments. Define 
$$
\begin{aligned}
	\sfmax = & \max_{1\le i\le n} \max_{1\le j\le p} \bE \suit{X_{i,j}-\mu_j}^4, \\
	\ssmax = & \max_{1\le i\le n} \max_{1\le j\le p} \bE \suit{X_{i,j}-\mu_j}^2, \\
	\ssmin = & \min_{1\le i\le n} \min_{1\le j\le p} \bE \suit{X_{i,j}-\mu_j}^2. 
\end{aligned}
$$
\begin{enumerate}[label=(C)]
	\item \label{assum:lyaponov} As $n\to\infty$, 
	$$
	\frac{ \sfmax }{\suit{ \ssmin }^2 } = o(np),  \quad \frac{ \ssmax }{ \ssmin  } = o(\sqrt{n}). 
	$$
\end{enumerate}

Condition~\ref{assum:lyaponov} imposes a restriction on the relative magnitudes of the second- and fourth-order moments of $X_{i,j}$. The first part ensures that the fourth moments are uniformly bounded relative to the squared second moments, while the second part limits the degree of variance heterogeneity. This technical condition  is used to validate   the Lyapunov condition for the central limit theorem  when handling the sum of $Y_i$'s.  

 Based on these assumptions, we establish the asymptotic normality of the test statistic as in the following theorem.

\begin{theorem}\label{Theorem:naive}
Assume that Conditions \ref{assum:mixing}, \ref{assum:var:low:bound}
and \ref{assum:lyaponov} hold. Then, under $H_0$,  as $n\to\infty$,
$$
T\stackrel{d}{\to} N(0,1). 
$$	
\end{theorem}

Based on Theorem \ref{Theorem:naive}, we reject the null hypothesis $H_0$ whenever $|T|> z_{1-\alpha/2}$, where $z_{1-\alpha/2}$ denotes the $1-\alpha/2$ quantile of the standard normal distribution. 

This testing procedure is straightforward and easy to implement.  Nevertheless, we remark that while this test is valid in terms of size, it is expected that it does not generate a high power. Essentially, we are testing an implication of $H_0$, $H'_0$, instead of the original hypothesis. We introduce a newly designed test that tests the hypothesis $H_0$ directly in the next subsection.

\subsection{The subsets-based test}\label{sec:improved}
  We consider a subsets-based approach. The main idea is as follows. For each subset of indices, we construct a test statistic as in Section \ref{sec:naive:test}. Then the final test combines all such test statistics across different subsets.
  
 Denote $\mathcal{S} = \set{1,\dots, p}$. Let $\mathcal{S}_{1},\dots, \mathcal{S}_d$ be subsets of $\mathcal{S}$, each of cardinality  $q\le p$. These subsets  may overlap; therefore $qd$  can be greater than $p$. 
 Define 
 $$
\begin{aligned}
	Y_{i}^{(\ell)} = \sum_{j\in \mS_{\ell}} X_{i,j} , \quad i=1,\dots,n. 
\end{aligned}
$$ 
We consider the test statistic 
$$
\mathcal{M} = \max_{1\le \ell \le d}\abs{T^{(\ell)}}.
$$
where
$$
T^{(\ell)} = \frac{\sum_{i=1}^n Y_i^{(\ell)}     }{\sqrt{n \widehat{\sigma}_{\ell\ell} }}, \quad \ell=1,\dots,d,
$$
with
$$
\begin{aligned}
		\widehat{\sigma}_{\ell\ell} =& \frac{1}{n}\sum_{i=1}^n \suit{Y_i^{(\ell)} -   \frac{1}{n}\sum_{i=1}^n  Y_i^{(\ell)}   }^2.
\end{aligned}
$$

The test statistic $\mM$ is designed to test whether the subsets pooled observations $Y_i^{(\ell)}$, $i=1,\dots, n, \ell=1,\dots, d$,  have  means of zero. Formally, this corresponds to testing a new null hypothesis $H_0^*$ as follows:
$$
 H_0^*: \sum_{j \in \mathcal{S}_{\ell} } \mu_j = 0, \quad \ell=1,\dots,d.
$$
The lemma below demonstrates that, by  choosing $p$ subsets appropriately,  this is equivalent to testing whether all $\mu_j = 0$,   i.e., $H_0$.

\begin{lemma}\label{lemma:main:subsets}
Choose $q<p$ such that,  $p$ and $q$ are coprime, i.e.,
 $gcd(p,q)=1.$ Define the subsets as follows:
 \begin{align*}
 	 \mathcal{S}_{\ell} = \set{\ell,\dots, \ell+q-1} , \quad \ell = 1,\dots, p.
 \end{align*}
 Here an index $i>p$ is interpreted as $i-p$.
 Then, $H_0$ holds if and only if $H_0^*$ holds.
\end{lemma}

Lemma \ref{lemma:main:subsets} shows that, by properly choosing $p$ subsets each containing $q$ elements, the new null hypothesis $H_0^*$ is equivalent to the original null hypothesis $H_0$. In practice, to enhance the power of the test, we select $d>p$ subsets. The first $p$ subsets are defined as in Lemma \ref{lemma:main:subsets}, while the remaining subsets can be specified by the user. In particular, users may decide which dimensions they wish to pool together related to the context.

Since the subsets $\mathcal{S}_1, \dots, \mS_d$ can be overlapped, the test statistics $T^{(\ell)}, \ell=1,\dots, d$ can be strongly dependent  which leads to challenges for  deriving the asymptotic distribution of $\mathcal{M}$.
 We  therefore resort to multiplier bootstrap to construct the threshold value of the test. The detailed test procedure is as follows.
 \begin{itemize}
 	\item Draw   i.i.d. observations $\xi_i\sim N(0,1)$, $i=1,\dots,n$.
 	\item Define
 	$$
 	\mathcal{M}_B = \max_{1\le \ell \le d}\abs{T_B^{(\ell)}},
 	$$
 	where
 	$$
 	  T^{(\ell)}_B =   \frac{\sum_{i=1}^n \xi_i Y_i^{(\ell)}    }{\sqrt{n \widehat{\sigma}_{\ell\ell} }}
 	$$
 	  \item Define $c_{\alpha}$ as the $1-\alpha$ conditional quantile of $\mathcal{M}_B$ given the data $(X_{i,1},  \dots, X_{i,p})$, $i=1,\dots,n$.
 	  \item Reject $H_0$ when $\mathcal{M}>c_{\alpha}$.
 \end{itemize}

To establish the asymptotic theory of the test procedure, we  again need a few conditions that are similar to but slightly stronger than those in Theorem \ref{Theorem:naive}.
 Condition \ref{assum:mixing} remains unchanged.
In addition, we impose a  new version of Condition \ref{assum:var:low:bound},  which are compatible for indices within the subsets.

\begin{enumerate}[label=(B$^\prime$)]
	\item \label{assum:joint:var:low:bound} There exists a constant $c>0$ such that, 
	$$
		\Var(Y_i^{(\ell)}) \ge c \sum_{j \in \mathcal{S}_{\ell}} \text{Var}(X_{i,j}),
	$$ 
	uniformly for all $1\le i\le n$ and $1\le \ell \le d$.
\end{enumerate}

% \begin{remark}
% A sufficient condition for Condition \ref{assum:joint:var:low:bound} to hold is the following. Let $\boldsymbol{K}_i$ denote the covariance matrix of $(X_{i,1}, \dots, X_{i,p})$. Then Condition \ref{assum:joint:var:low:bound} is satisfied if 
%$$
%\min_{1\le i\le n}\lambda_{\min}(\boldsymbol{K}_i)\ge c \ssmax,
%$$
%for some $c>0$. Here $\lambda_{\min}(\boldsymbol{K}_i)$ denotes the minimum eigenvalue of the matrix $\boldsymbol{K}_i$.  	
% \end{remark}
 
 We also impose
a restriction on the relative magnitudes of the second- and fourth-order moments of $X_{i,j}$.  This condition is stronger than condition \ref{assum:lyaponov} owing to the  fact that we are dealing with high dimensional subsets. 

 \begin{enumerate}[label=(C$^\prime$)]
	\item \label{assum:joint:sequence:lyaponov} As $n\to\infty$,   
\begin{equation*}
\begin{aligned}
 \frac{\sfmax }{\suit{\ssmin}^2 } &=o(1)\frac{nq}{\log^{5}(pn) },  \quad   \frac{\ssmax }{\ssmin} = o(1)\min\set{\suit{\frac{n}{\log^5(pn) } }^{1/2}, \log^2 p } .
\end{aligned}
\end{equation*}
\end{enumerate}

Moreover, we impose some conditions on the choices of $q$ and $d$.  Most notably, these conditions require that $q$ and $d$ are bounded from above by some power function of $p$. In addition $q$ is bounded from below, related to $p$, $n$ and $\ssmin$.

 \begin{enumerate}[label=(D)]
	\item \label{assum:joint:sequence:highdim} As $n\to\infty$,  $ q<p$,  $\log q = O(1)\log p$, $\log d = O(1)\log p$,   and 
\begin{equation*}
\begin{aligned}
\frac{n}{\log^{5} (pn)}  \to\infty, \quad 
 	\frac{nq\ssmin }{\log^{9} (pn) } \to \infty, 
\end{aligned}
\end{equation*}
\end{enumerate}

\begin{theorem}\label{theorem:size:improved}
Assume that Conditions  \ref{assum:mixing},        \ref{assum:joint:var:low:bound}, \ref{assum:joint:sequence:lyaponov} and \ref{assum:joint:sequence:highdim} hold. Moreover, assume that, $\max_{1\le i\le n}\max_{1\le j\le p}\abs{X_{i,j}}\le C$ for some constant $C>0$.
  Then, under $H_0$, as $n\to\infty$,
$$
\Pr(\mathcal{M}>c_{\alpha})\to \alpha.
$$	
\end{theorem}

Theorem \ref{theorem:size:improved} establishes that, the  proposed testing procedure  based on multiplier bootstrap is consistent.  
Next, we analyze the power of the proposed test procedure. We consider the following local  alternative hypothesis  
$$
H_1:  \max_{1\le \ell \le d}  \frac{ \abs{\sum_{j\in \mS_{\ell}}  \mu_j} }{\sqrt{  q\ssmax }  } \ge \frac{\sqrt{\lambda\log d }}{\sqrt{n}},
$$
with $\lambda\to\infty $ as $n\to\infty$.

Recall that $\ssmax$ is the maximum variance of $X_{i,j}$ across all $i$ and $j$, resulting in a potential maximum variance of $Y_i^{\ell}$ at the level $q\ssmax$. The local alternative hypothesis $H_1$ requires that the maximum normalized subset mean deviation exceeds $\sqrt{\lambda\log d}/\sqrt{n}$, a threshold tending to zero. This condition ensures that at least one of the subset means is sufficiently large to be detected by the test as the sample size increases. The threshold $\sqrt{\lambda\log d}/\sqrt{n}$ reflects the balance between the sample size $n$ and the number of subsets $d$.

\begin{theorem}\label{theorem:power:improved}
Assume that Conditions  \ref{assum:mixing},        \ref{assum:joint:var:low:bound}, \ref{assum:joint:sequence:lyaponov} and \ref{assum:joint:sequence:highdim} hold.  Moreover, assume that, $\max_{1\le i\le n}\max_{1\le j\le p}\abs{X_{i,j}}\le C$ for some constant $C>0$. Then, under $H_1$, as $n\to\infty$,
$$
\Pr(\mathcal{M}>c_{\alpha})\to 1.
$$	
\end{theorem}

Theorem \ref{theorem:power:improved} establishes that,  under the local alternative hypothesis $H_1$, the power of the test,  converges to 1.

\section{Application: Backtesting for VaR}\label{sec:backtest}
In this section, we demonstrate how the proposed test procedure in Section \ref{sec:method} can be applied to the backtesting of various models in value-at-risk (VaR) forecasting.

Let $U_{i,j}$ denote the negative log-return of asset $j$ at time $i$, and let $R_{i,j}$ and $R^*_{i,j}$ be two competing forecasts of the $1-\theta_{j,0}$ conditional quantile of asset $j$ at time $i$, respectively. We focus on two types of backtests commonly used in practice:
(1) validation backtests, which assess whether one VaR forecast model is correctly specified, and
(2) comparative backtests, which compare the predictive accuracy of two  forecast models.

\subsection{Validation Backtest}

The purpose of the validation backtest is to determine whether the forecast 
$R_{i,j}$ (or $R_{i,j}^*$) provides  an accurate prediction 
of the target   $1-\theta_{j,0}$ conditional quantile for  asset $j$ at time $i$. 
Define  an exceedance indicator
$$
\widetilde{X}_{i,j} = \mI\suit{ U_{i,j} > R_{i,j}}, \quad i=1,\dots,n, \ j=1,\dots, p.
$$
Under a correctly specified model, the random vectors
$$
(\widetilde{X}_{i,1}, \dots, \widetilde{X}_{i,p}), \quad i=1,\dots, n,
$$
are i.i.d., and each component follows a Bernoulli distribution with parameter $\theta_j$,  where $\theta_j$ is the probability that the loss exceeds the VaR forecast for asset $j$.
Thus, the null hypothesis  to test is 
$$
H_{0,\text{VaR}}: \quad \theta_j = \theta_{j,0}, \quad j=1,\dots,p.
$$
Here, $\theta_{j,0} = \theta_{j,0}(n) \to 0$ as $n\to\infty$, and $n\theta_{j,0} = O(1)$.

The standard backtest for VaR specified in the Basel regulation (\cite{basel2013bank}, pages 103-108) uses the test statistic $\sum_{i=1}^n \widetilde{X}_{ij}$ to test the null hypothesis $\theta_j = \theta_{j,0}$. However,  as discussed above, when $\theta_{j,0}$ is low, i.e., $\theta_{j,0} = \theta_{j,0}(n) \to 0$ and $n\theta_{j,0} = O(1)$, this test statistic is no longer asymptotically normal. In this case, we can use the subsets-based test procedure as follows.
Define
$$
X_{i,j} = \widetilde{X}_{i,j} - \theta_{j,0}.
$$
By reformulating the null hypothesis $H_{0,\text{VaR}}$ into the mean zero test problem
$$
\bE(X_{i,j}) = 0, \quad j=1,\dots,p,
$$
by contrast, we can use the test procedures detailed in Section \ref{sec:method}.

In practice, the target exceedance probability is typically set identically across all
$p$ assets, i.e., $\theta_{1,0} = \dots = \theta_{p,0} \equiv \theta_0$. 
 Under the null $H_{0,\text{VaR}}$, Condition \ref{assum:joint:sequence:lyaponov} and  \ref{assum:joint:sequence:highdim}  reduce to
\begin{equation}\label{condition:validation}
\frac{nq\theta_0}{\log^9(pn)}\to\infty,  \quad  \frac{\log^5 p}{n} \to 0.	
\end{equation}
Then, we have the following corollary. 
\begin{corollary}
	 Assume that Conditions  \ref{assum:mixing},        \ref{assum:joint:var:low:bound} hold. Moreover, assume that \eqref{condition:validation} holds.  Then, under $H_0$, we have that,  as $n\to\infty$, 
	 $$
\Pr(\mathcal{M}>c_{\alpha})\to \alpha.
$$	
\end{corollary}

Under the alternative, we have that, 
$$
\begin{aligned}
	\ssmin  = & \min_{1\le j\le p}  \theta_j(1-\theta_j)  \asymp \min_{1\le j \le p} \theta_j=:\theta_{\min}, \\
	\ssmax = &\max_{1\le j\le p}  \theta_j(1-\theta_j) \asymp \max_{1\le j \le p} \theta_j=:\theta_{\max}, \\
	\sfmax = & \max_{1\le j\le p}  \theta_j(1-\theta_j)(\theta_j^3+(1-\theta_j)^3)\asymp  \theta_{\max}. \\
\end{aligned}
$$
Then,   Condition \ref{assum:joint:sequence:lyaponov}    reduces to
\begin{equation}\label{eq:backtest:conditionC}
	 \frac{\theta_{\max} }{\suit{\theta_{\min}}^2 } =o(1)\frac{nq}{\log^{5}(pn) },  \quad   \frac{\theta_{\max} }{\theta_{\min}} = o(1)\min\set{\suit{\frac{n}{\log^5(pn) } }^{1/2}, \log^2 p } .
\end{equation}
Similarly, Condition \ref{assum:joint:sequence:highdim} specializes to
\begin{equation}\label{eq:backtest:conditionD}
\begin{aligned}
\frac{n}{\log^{5} (pn)}  \to\infty, \quad 
 	\frac{nq\theta_{\min} }{\log^{9} (pn) } \to \infty, 
\end{aligned}
\end{equation}
The local alternative $H_1$ can be written as 

$$
H_1: \max_{1\le \ell \le d} \abs{ \frac{1}{q}\sum_{j\in \mathcal{S}_{\ell}} \frac{\theta_j-\theta_{j,0} }{\theta_{\max}}   } \ge  \frac{\sqrt{ \lambda \log d}}{\sqrt{nq\theta_{\max} }}.
$$
Here, $(\theta_j-\theta_{j,0})/\theta_{\max}$ represents the percentage deviation of the exceedance probability $\theta_j$ from the target $\theta_{j,0}$. The local alternative hypothesis $H_1$ requires that the maximum average percentage deviation across subsets exceeds a threshold, ensuring that at least one subset has a sufficiently large deviation from the null to be detected by the test as the sample size increases.

\begin{corollary}

Assume that Conditions  \ref{assum:mixing},        \ref{assum:joint:var:low:bound}, \eqref{eq:backtest:conditionC} and \eqref{eq:backtest:conditionD} hold.   Then, under $H_1$, as $n\to\infty$,
$$
\Pr(\mathcal{M}>c_{\alpha})\to 1.
$$	
\end{corollary}

\subsection{Comparative Backtest}
When both VaR forecast models are not rejected by the validation backtest, each can be regarded as a statistically reliable predictor of the corresponding VaRs. A natural next step is to compare their relative forecasting performance. Comparative backtesting serves precisely this purpose by evaluating whether one forecast model systematically outperforms another in terms of quantile forecast accuracy, as follows.

Define the scoring function
$$
S(r,x,\theta) = \suit{1-\theta - \mI\suit{x > r}} G(r) + \mI\suit{x > r} G(x),
$$
where $G$ is an increasing function on $\mathbb{R}$. By \cite{Saerens2000build} and \cite{nolde2017elicitability}, this scoring rule is strictly consistent  for eliciting the $\theta-$quantile, i.e.,   in expectation it is minimized when $r$  equals the true $\theta$-quantile.   Thus, a better prediction should have a lower score, which provides a valid basis for comparative backtesting.

The null hypothesis for the comparative test is formulated as 
$$
H_{0,\text{Com}}: \quad \bE \suit{ S(R_{i,j}, U_{i,j},\theta_{j,0} ) - S(R^*_{i,j}, U_{i,j}, \theta_{j,0}) } = 0, \quad i=1,\dots,n, \ j=1,\dots,p.
$$
Define
$$
X_{i,j} = S(R_{i,j}, U_{i,j}, \theta_{j,0}) - S(R^*_{i,j}, U_{i,j},\theta_{j,0}).
$$
Then the null hypothesis can be regarded equivalently as
$$
\bE(X_{i,j}) = 0, \quad j=1,\dots,p.
$$
Note that, 
$$
\begin{aligned}
	X_{i,j} =& (1-\theta_{j,0} ) \suit{ G( R_{i,j}) -G(R_{i,j}^*) } +\mI\suit{ U_{i,j}>R_{i,j} }\suit{ G(U_{i,j})-G(R_{i,j})  } \\
			&-  \mI\suit{ U_{i,j}>R_{i,j}^* }\suit{ G(U_{i,j})-G(R_{i,j}^*)  }.
\end{aligned}
$$

Assume that both forecasts  are accurate.  When the target quantile level $\theta_{j,0}$ is small, i.e., $\theta_{j,0} = \theta_{j,0}(n) \to 0$ and $n\theta_{j,0} = O(1)$, the two indicator terms activate with probability $\theta_{j,0}\to0$ and the difference
$ G( R_{ij}) -G(R_{ij}^*)$ shrinks to zero. Consequently,  
the random variable $X_{i,j}$ shrinks to zero as $n\to\infty$.  
 The standard comparative test based on normal approximation \citep{nolde2017elicitability} again becomes invalid.

 The procedure outlined in Section \ref{sec:method} can instead be applied to test $H_{0,\text{Com}}$. In practice, the joint distribution of $R, R^*,$ and $U$ can be complex, which prevents a straightforward way to validate the general conditions  \ref{assum:joint:sequence:lyaponov} and \ref{assum:joint:sequence:highdim}.
 To ensure the boundedness condition required by Theorem \ref{theorem:size:improved}, we select a bounded, increasing function $G$. For all numerical experiments, we adopt the logistic function
$$
G(x) = \frac{1}{\exp(-x)+1}.
$$

\section{Simulation}\label{sec:simulation}

In this section, we present a simulation study to illustrate the finite sample performance of our testing procedures.   
The samples are generated from the following models.

Define $\boldsymbol{\Sigma}_1$ as a $p\times p$ matrix with elements $(\sigma_{ij})_{1\le i\le p, 1\le j\le p}$, for $\sigma_{ii} =1$, $\sigma_{ij} = 0.7$ for  $(i,j) \in \{(2k-1,2k),(2k,2k-1)\}: k =1,2, \dots, \lfloor p/2\rfloor\}$ and $\sigma_{ij}=0$ otherwise.  
 Define $\boldsymbol{\Sigma}_2$ as a $p\times p$ matrix with elements $(\sigma_{ij})_{1\le i\le p, 1\le j\le p}$, where $\sigma_{ij} = 0.5^{|i-j|}$. Define $z_\alpha$ as the $\alpha$ quantile of a standard normal distribution. 

 We start by considering two data generating processes that resemble the setting of the validation 
backtest, where the observations are binary.

{\bf \noindent Model (A1)}.  We generate $(Z_{i,1},\dots, Z_{i,p})$, $i=1,\dots,p$ independently from $N(0,\boldsymbol{\Sigma}_1)$.  For $i=1,\dots,n$, $j=1,\dots, p$,  define 
$$
X_{i,j} =  \mI\suit{ Z_{i,j}>z_{1-\theta_j} }-0.01.
$$ 
Under the null hypothesis, $\theta_j = 0.01, j=1,\dots,p$.    Under the alternative hypothesis, a subset of the parameters deviates from this value. Specifically, among the $p$ dimensions, there are $p_0$ dimensions with $\theta_j$'s increased to $0.025$. Another $3p_0$ dimensions have a reduced $\theta_j$'s at $0.005$. The remaining $p-4p_0$ dimensions have unchanged $\theta_j$'s at $0.01$. In such a setting, while the $\theta_j$'s deviates from $0.01$ at $4p_0$ dimensions, their overall average remains at $0.01$.

{\bf \noindent Model (A2)}.    We generate $(Z_{i,1},\dots, Z_{i,p})$, $i=1,\dots,p$ independently from $N(0,\boldsymbol{\Sigma}_2)$. Other settings are the same as in Model (A1).

 Next, we consider two data generating processes with continuous observations, which resembles the setting of the comparative backtest. Simulations based on such general models demonstrate the generality of our testing procedure.

{\bf \noindent Model (B1)}. We generate $(Z_{i,1},\dots, Z_{i,p})$, $i=1,\dots,p$ independently from $N(0,\boldsymbol{\Sigma}_1)$.   For $i=1,\dots,n$, $j=1,\dots, p$,  define  
$$
X_{i,j} = g(\Phi(Z_{ij}))  - \mu_j, 
$$
where
$$
g(x) = \begin{cases}
	\frac{2\alpha_n}{1-\alpha_n}x - \alpha_n, & \text{if} \ 0\le x\le 1-\alpha_n, \\
	\frac{2}{\alpha_n} x + 1-	\frac{2}{\alpha_n} , & \text{if}\ 1-\alpha_n <x \le 1, 
\end{cases}
$$
with $\alpha_n = 0.01$. 
By such transformation, $g(\Phi(Z_{ij}))$ follows a mixture distribution $(1-\alpha_n) U(-\alpha_n, \alpha_n) +\alpha_n U(-1,1)$. Here, $U(a,b)$ denotes the uniform distribution on the interval $(a,b)$.  That means $X_{i,j}$ has a mean $-\mu_j$.

Under the null hypothesis,  we set $\mu_j = 0, j=1,\dots,p$.  Under the alternative hypothesis, a subset of the parameters deviates from this value. Specifically,  for $p_0$ dimensions, we set their $\mu_j$'s  to $-0.0075$, while for another $p_0$ dimensions, we set their $\mu_j$'s to $0.0075$. The remaining $p-2p_0$ parameters remain unchanged at $0$.  In such a setting, among the $p$ dimensions, there are $2p_0$ parameters deviating from the null, while the overall average of the parameters remains at $0$.

{\bf \noindent Model (B2)}.    We generate $(Z_{i,1},\dots, Z_{i,p})$, $i=1,\dots,p$ independently from $N(0,\boldsymbol{\Sigma}_2)$. Other settings are the same as in Model (B1).

We generate data from the above four  data generating processes with $n = 500, p = 100$ and $p_0 = 20$. We compare our test procedures, {\it the naive test}  and {\it the subsets-based pooling test},   with the conventional maximum-type mean test without pooling \citep{chang2017simulation}, which we refer to as the {\it marginal test}. The detailed procedure is described below.

 \begin{itemize}
 \item Consider the test statistic
 $$
\mathcal{M}_{m} = \max_{1\le j \le p}\abs{  \frac{\sum_{i=1}^n X_{i,j}}{\set{\suit{\sum_{i=1}^n (X_{i,j} - \frac{1}{n}\sum_{i=1}^n X_{i,j} }^2}^{1/2} } }. 
$$
 	\item Draw i.i.d. observations $\xi_i\sim N(0,1)$, $i=1,\dots,n$.
 	\item Define 
 	$$
 	\mathcal{M}_{m,B}  =  \max_{1\le j \le p}\abs{  \frac{\sum_{i=1}^n \xi_i X_{i,j}}{\set{\suit{\sum_{i=1}^n (X_{i,j} - \frac{1}{n}\sum_{i=1}^n X_{i,j} }^2}^{1/2} } }. 
 	$$ 
 	  \item Define $c_{m,\alpha}$ as the $1-\alpha$ conditional quantile of $\mathcal{M}_{m,B}$ given the data $(X_{i,1},   \dots, X_{i,p})$, $i=1,\dots,n$.
 	  \item Reject  $H_0$ when  $\mathcal{M}_{m}>c_{m,\alpha}$.
 \end{itemize}

 First, we fix the  value of $d = 2p$ and vary the subsets size
  $$q \in \{9,19,29,39,49,59,69,79,89\}.$$ 
 These choices of $q$ ensure that $gcd(p,q)=1$.  The first $p$ subsets are chosen according to Lemma~\ref{lemma:main:subsets}. The remaining $d-p$  subsets are generated by randomly sampling $q$ distinct elements from the index set $\set{1,\dots, p}$.

 The results are presented in Figure \ref{fig:q}. We observe that the marginal test fails. This can be attributed to the failure of asymptotic normality for marginal dimensions. The naive test performs well in terms of empirical size because it pools all observations, making the effective sample size sufficiently large. However, it lacks power when marginal means with opposite signs cancel out in the average, as is the case in our simulation setting. 

The choice of  $q$ involves a size-power trade-off for our proposed subsets-based pooling test. Increasing $q$ enlarges the sample size within each subset, improving size accuracy. However, setting $q$ too large can again lead to the problem of marginal means with opposite signs canceling out in the average.

 Next, we fix the value of $q=49$ and vary the values of $d$ from $100$ to $500$. 
 The results are presented in Figure \ref{fig:d}. 
 We observe that, within this range, the empirical size remains largely unchanged, while increasing $d$ may improve the power of the test. However, a larger $d$ also increases computational cost, and excessively large values may violate Condition \ref{assum:joint:sequence:highdim}. Based on our simulation study, we recommend choosing $d = \beta p$ with $\beta \in [2,5]$.

\begin{figure}[htbp]
\centering
\begin{subfigure}{1\textwidth}
\includegraphics[width=\linewidth]{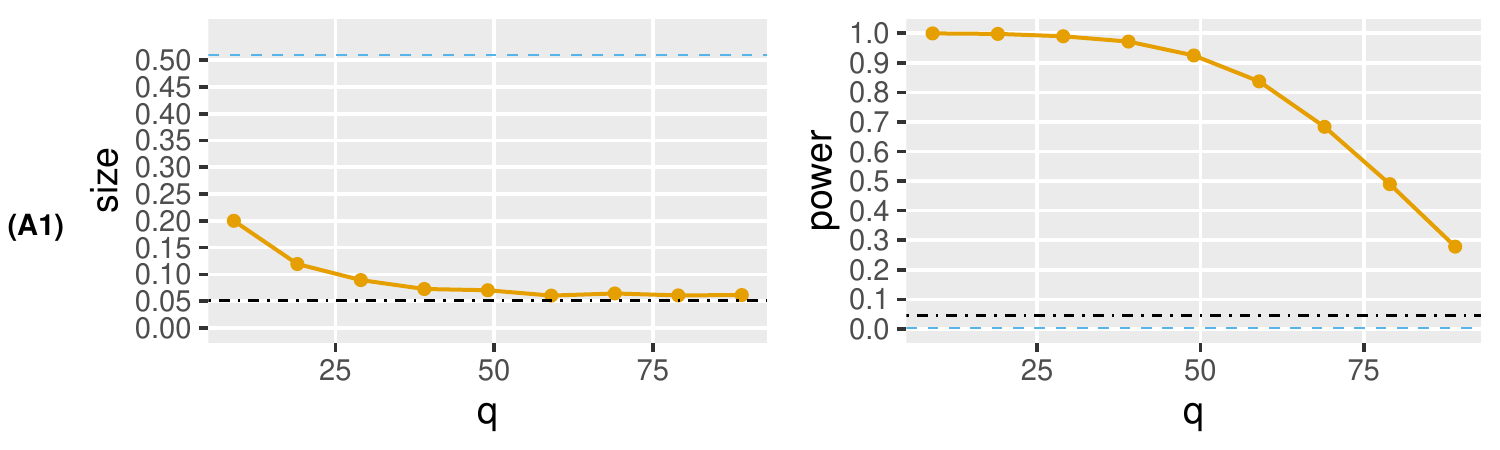}
\end{subfigure}
\begin{subfigure}{1\textwidth}
\includegraphics[width=\linewidth]{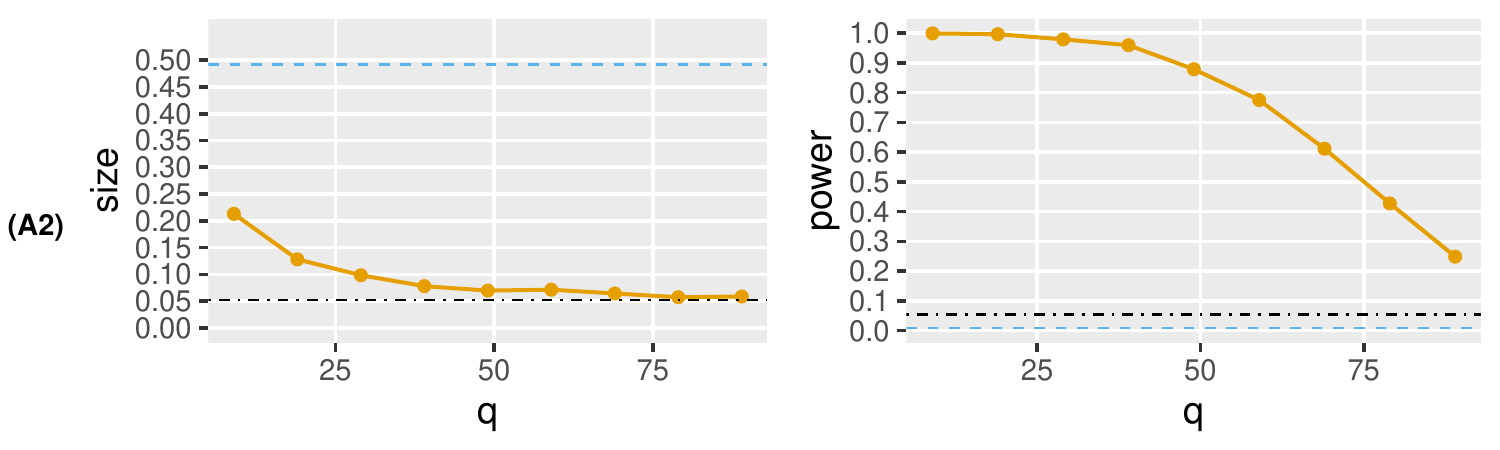}
\end{subfigure}
\begin{subfigure}{1\textwidth}
\includegraphics[width=\linewidth]{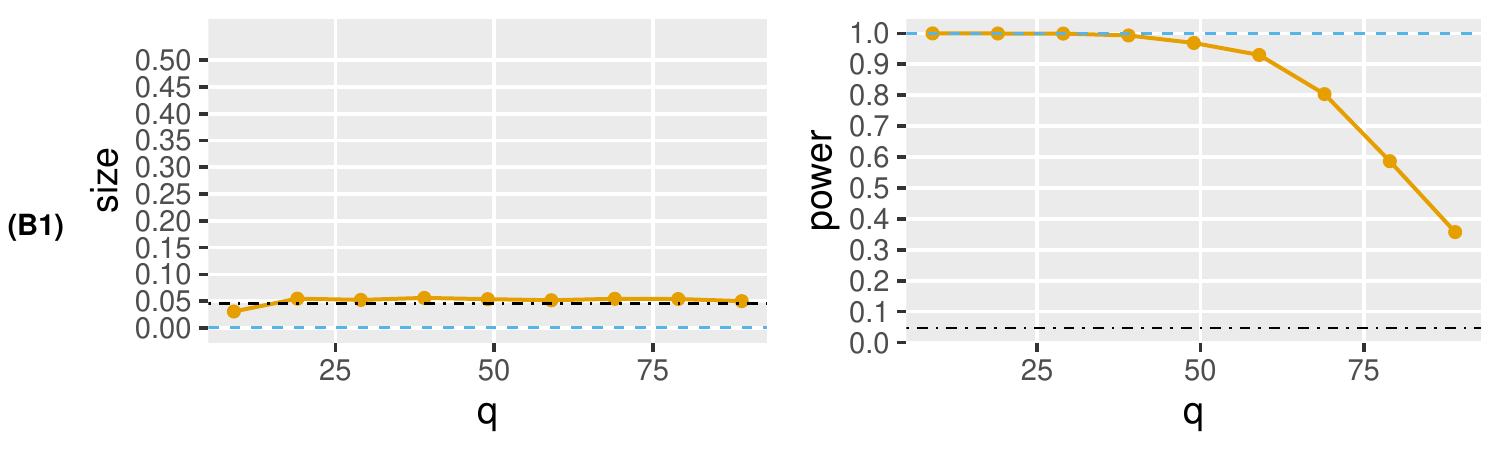}
\end{subfigure}
\begin{subfigure}{1\textwidth}
\includegraphics[width=\linewidth]{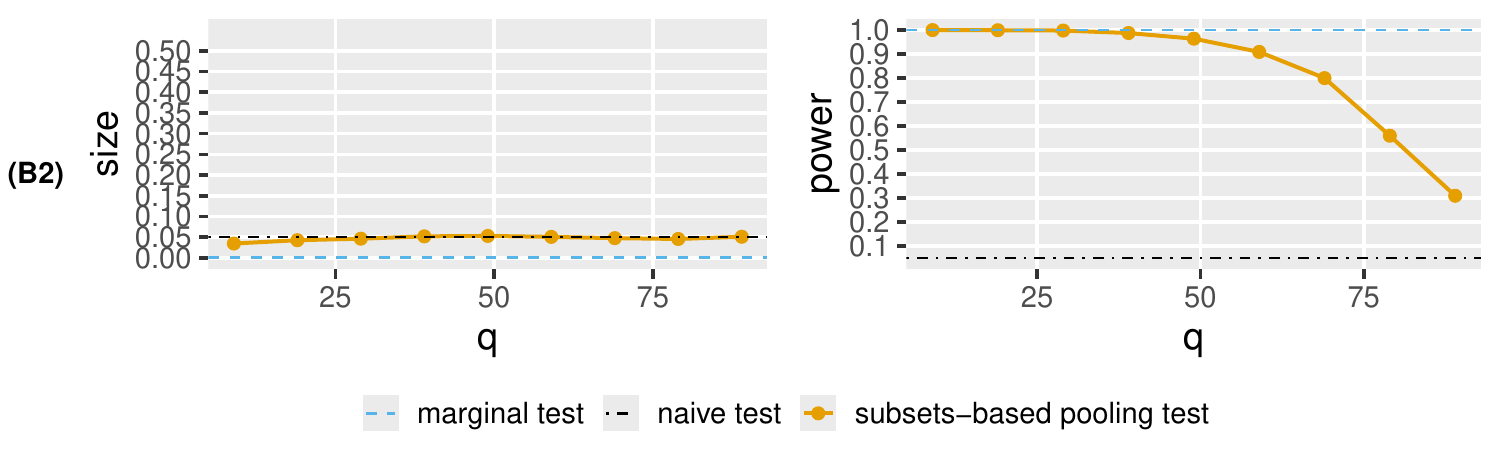}
\end{subfigure}
\caption{The empirical size (left) and power (right) of subsets-based pooling test, marginal test, 
and naive test  over different levels of $q$. }
\label{fig:q}
\end{figure}

\begin{figure}[htbp]
\centering
\begin{subfigure}{1\textwidth}
\includegraphics[width=\linewidth]{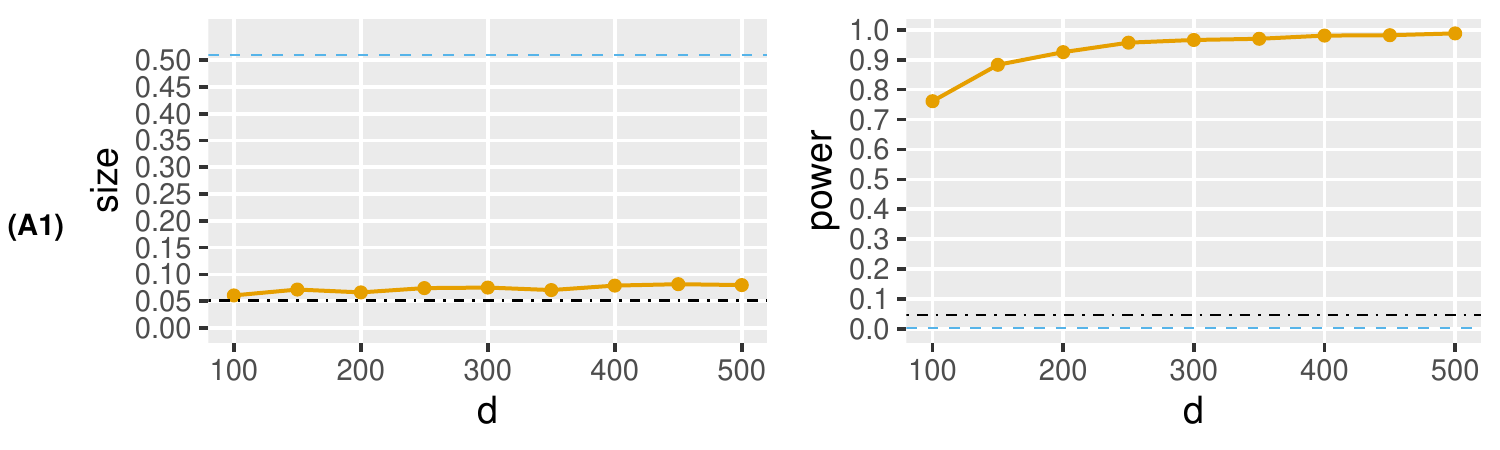}
\end{subfigure}
\begin{subfigure}{1\textwidth}
\includegraphics[width=\linewidth]{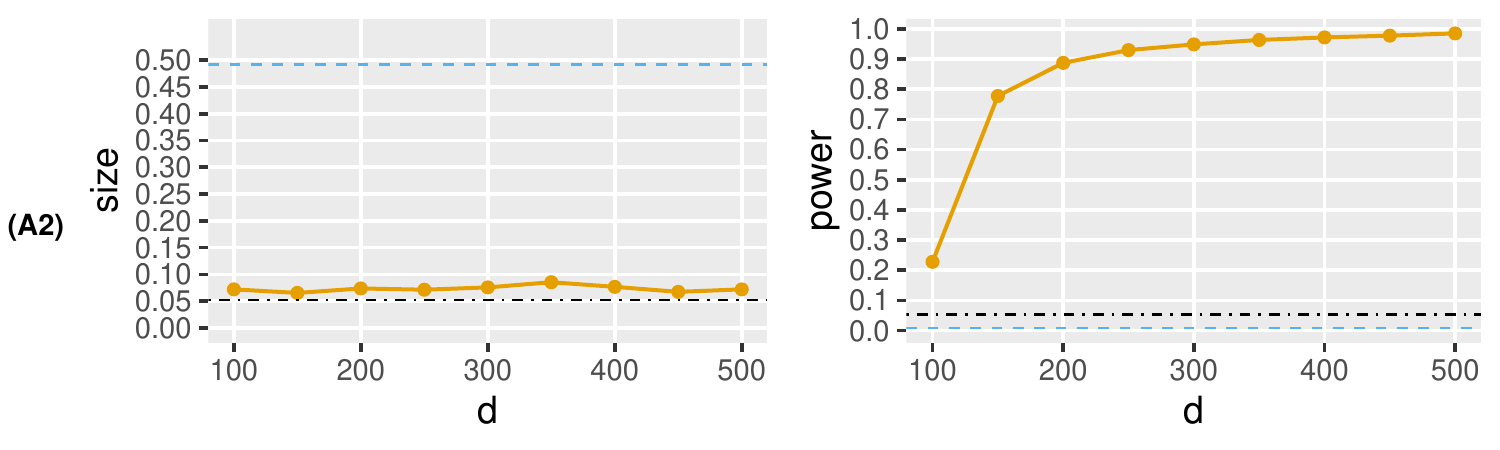}
\end{subfigure}
\begin{subfigure}{1\textwidth}
\includegraphics[width=\linewidth]{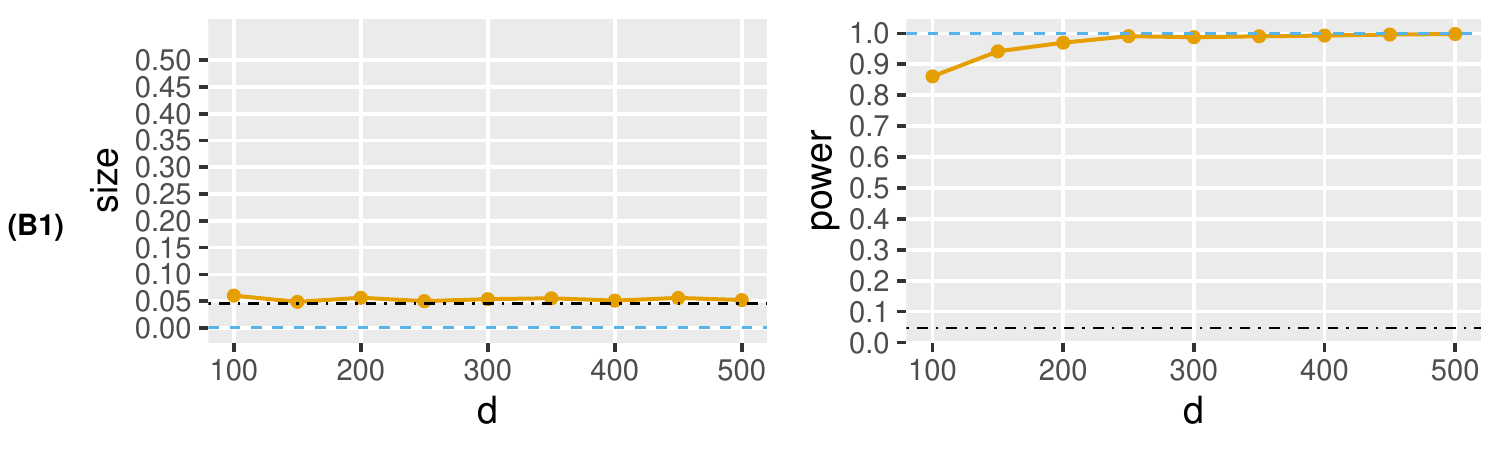}
\end{subfigure}
\begin{subfigure}{1\textwidth}
\includegraphics[width=\linewidth]{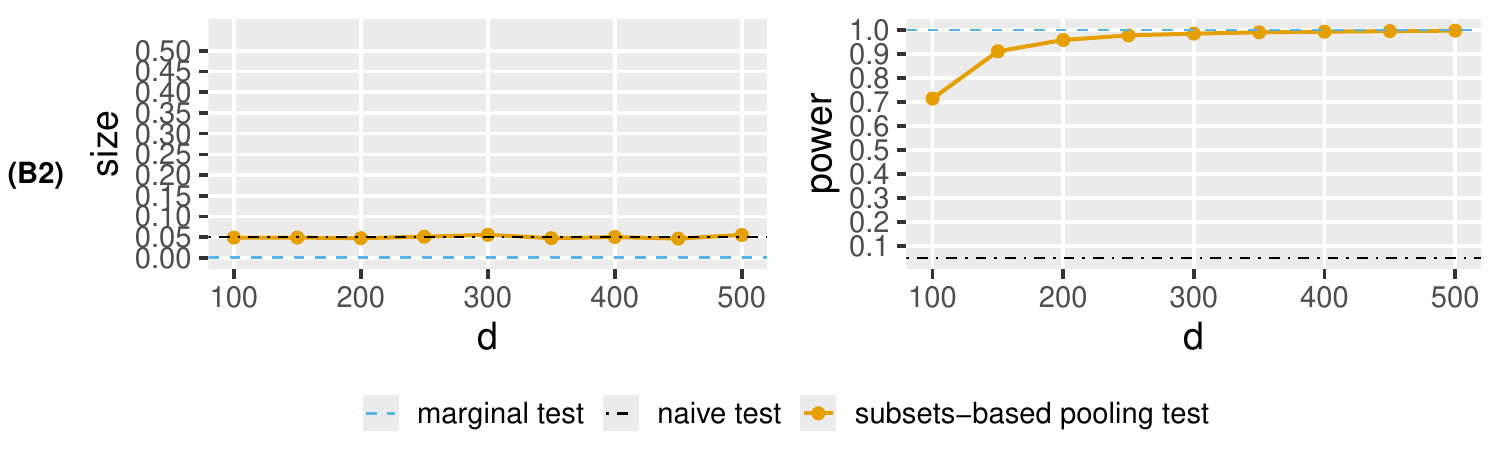}
\end{subfigure}
\caption{The empirical size (left) and power (right) of subsets-based pooling test, marginal test, 
and naive test  over different levels of $d$. }
\label{fig:d}
\end{figure}

\section{Real data Application} \label{sec:application}

In this section, we illustrate the real data application of the proposed methods in a backtesting context. We focus on the constituents of the S\&P 100 index and
consider a two-year backtesting period from January 1, 2023, to December 31, 2024. Over this horizon, we obtain
 $n=502$ daily observations of negative log returns (losses) for each constituent used in the backtesting.

We use various VaR forecasting methods to forecast the one-day-ahead conditional VaR on each day in the backtesting period, and then compare their performance. To conduct conditional VaR forecasting, for each constituent, the dynamics of losses are modeled using an AR(1)-GARCH(1,1) time series model \citep{mcneil2000estimation}. We employ a rolling estimation window consisting of 3,000 observations prior to the backtesting period. Consequently, each constituent must have at least 3,000 historical observations prior to January 1, 2023. Constituents not meeting this criterion are excluded, leaving $p=92$ constituents for the analysis. 

Let $U_{i,j}$ denote the negative log-returns at time $i$ for constituent $j$. For each $j$, 
the AR(1)-GARCH(1,1) model is defined as 
\begin{align*}
	U_{i,j} = \mu_{i,j} +\sigma_{i,j} Z_{i,j}, \quad
	\mu_{i,j} = \alpha_{0,j} + \alpha_{1,j} U_{i-1,j},\quad  
	 \sigma_{i,j}^2 = \beta_{0,j} +\beta_{1,j}  \sigma_{i-1,j}^2Z_{i-1,j}^2+\beta_{2,j} \sigma_{i-1,j}^2,
\end{align*}
where the innovations $Z_{i,j}$ are i.i.d. with zero mean and unit variance. The parameters are estimated by maximum likelihood assuming a standardized skew-t distribution for $Z_{i,j}$ \citep{fernandez1998bayesian}. Nevetheless, we do not necessarily making use of such a distributional assumption. Instead, using the estimated conditional means and volatilities, we obtain the filtered residuals
$$
\widehat{Z}_{i,j} = \frac{U_{i,j} -\widehat{\mu}_{i,j} }{\widehat{\sigma}_{i,j} },
$$
which serve as proxies for the realized innovations. The one-step-ahead dynamic VaR forecasts at the quantile level $1-\theta_0$  are then given by 
$$
\widehat{\text{VaR}}_{i,j}(\theta_0) = \widehat{\mu}_{i,j}  + \widehat{\sigma}_{i,j} \widehat{\text{VaR}}_{Z,j}(\theta_0),
$$
where $\widehat{\text{VaR}}_{Z,j}(\theta)$ denotes the estimated VaR of the innovations $Z_{i,j}$. Throughout the empirical application, we set $\theta_0 =0.01$.
We consider three different approaches for the estimation  of $\widehat{\text{VaR}}_{Z,j}(\theta_0)$.

\begin{itemize}
	\item[(1)] {\it Empirical method}: VaR is estimated based on the empirical quantiles of the filtered residuals.
 \item[(2)] {\it Standardized skew-t (SSTD) method}: VaR is computed from the quantiles of the standardized skew-t distribution fitted to the filtered residuals. This method is a parametric approach that relies on the distributional assumption of the innovations.
 \item[(3)] {\it Extreme value theory (EVT) method}: VaR is estimated using extreme value statistics with the number of tail observations setting to $50$, see Chapter 4.3 of \cite{haan2006extreme} for details. 
  \end{itemize}

Before applying the proposed methodology to the data, we first check the appropriateness of several underlying assumptions.  We start with Condition \ref{assum:mixing}, which is the key requirement for our analysis. A direct test for  Condition \ref{assum:mixing} may not be  feasible in practice.  Nevertheless, in the 
context of validation test, the dependence among the indicators is 
driven by the extreme events of the losses exceeding corresponding 
marginal quantiles. Therefore, evaluating tail dependence among the 
residuals provides useful insight regarding the validity of this condition.
 Specifically, we compute the upper tail dependence coefficients \citep{sibuya1960bivariate} between the filtered residuals $\widehat{Z}_{i,j}$ from the AR(1)-GARCH(1,1) models,
$$
\lambda_{j_1,j_2} = \frac{1}{nu}\sum_{i=1}^n \mI\suit{\widehat{F}_{n,j_1}(\widehat{Z}_{i,j_1})> u,\widehat{F}_{n,j_2}(\widehat{Z}_{i,j_2})>u   },
$$
where $\widehat{F}_{n,j}$ denotes the empirical cumulative distribution function of $\widehat{Z}_{1,j},\dots, \widehat{Z}_{n,j}$. The threshold $u$ is a 
  small positive constant,  which we set to $u=0.01$ reflecting our choice of $\theta_0$.

 Figure \ref{figure:heatplot} shows that tail dependence is generally weak across different sectors, and each constituent exhibits strong tail dependence only with a small number of other constituents. Based on these observations, we conclude that the tail dependence within this dataset has a ``block wise'' structure. This supports the $m$-dependent structure required by Condition \ref{assum:mixing}.

Recall that Condition \ref{assum:joint:var:low:bound} requires the absence of strong multicollinearity. We argue that this condition is generally easy to satisfy in practice. With assuming market efficiency, there is no arbitrage opportunities in financial markets. Therefore, perfect or near-perfect collinearity  should not exist in financial markets.  We do not further check the technical conditions in Condition \ref{assum:joint:sequence:lyaponov} and \ref{assum:joint:sequence:highdim}.

\begin{figure}[htbp]
	\centering
		\includegraphics[width = 1\textwidth]{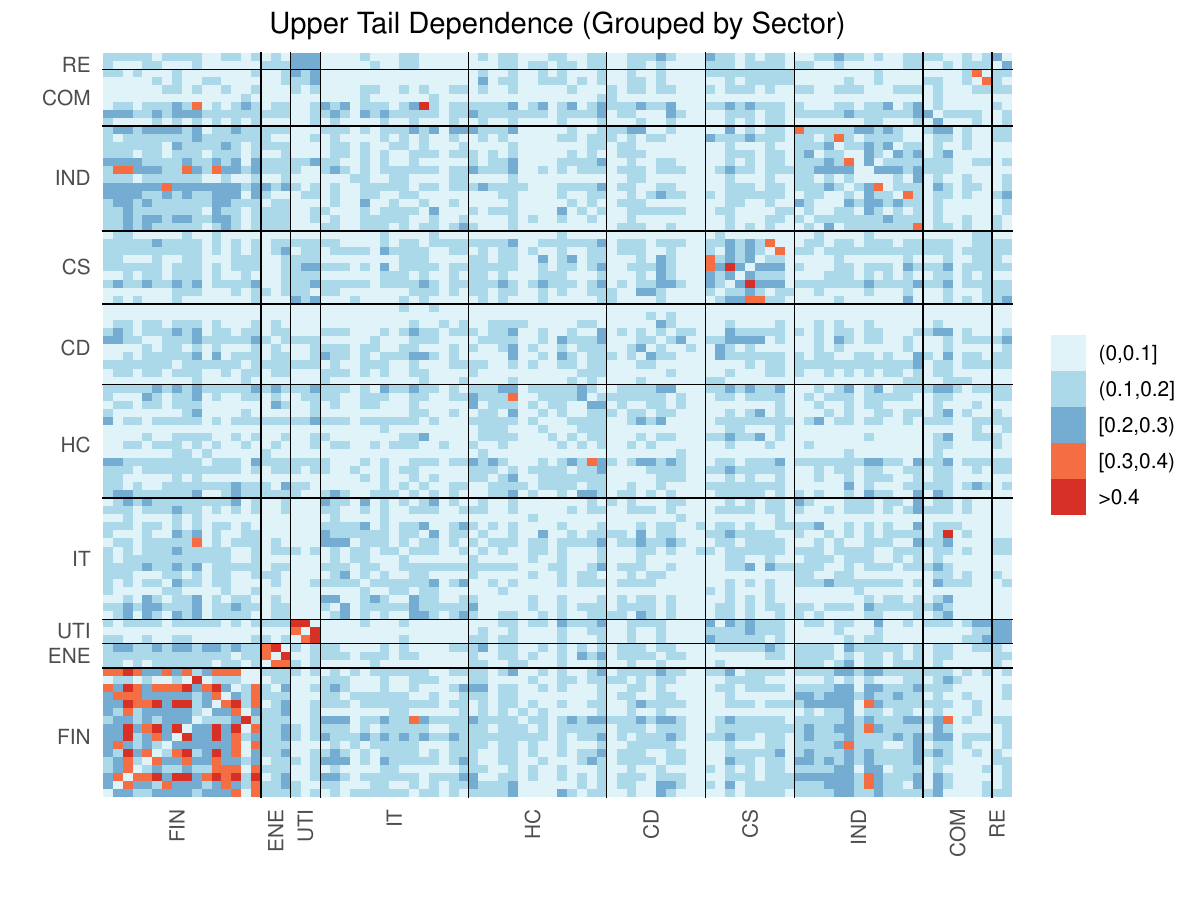}
	\caption{Heatmap of the upper tail dependence coefficients for the filtered residuals in the first training period for the  S\&P constituents. Each row and column corresponds to a stock, and stocks are grouped by the Global Industry Classification Standard (GICS)  employed by Standard and Poor's: Financials (FIN), Energy (ENE), Utilities (UTIL), Information Technology (IT), Health Care (HC), Consumer Discretionary (CD), Consumer Staples (CS), Industrials (IND), Communication Services (COM), and Real Estate (RE). }
	\label{figure:heatplot}
\end{figure}

We then conduct both validation tests and comparative backtests for the three VaR forecasting methods. We set the parameters $q=49$ and $d=300$ for the proposed tests.  The subsets are constructed in the same way as in Section \ref{sec:simulation}.
 Table \ref{tab:pvalues} shows the $p-$values of the validation test in the diagonal element for each method respectively, and presents those of the comparative tests in the off-diagnoal elements.  The $p$-values in the diagonal indicate that all three methods pass the validation tests, concluding that all methods provide reliable VaR estimates. By contrast, the comparative tests are significant, with $p$-values below 0.05, allowing us to rank the methods: the SSTD method performs best, followed by the EVT method as a semi-nonparametric method, and then the empirical method. 
 
The superior performance of the SSTD method can be attributed to the 
goodness-of-fit when modeling the residuals by the skew-t model. 
Compared to the SSTD method, the EVT method imposes parametric 
assumptions only on the tail of the distribution, leading to relatively 
higher estimation uncertainty and thus the relatively worse performance 
in VaR forecasting. Nevertheless, these assumptions are consistent with 
the distribution of the residuals, particularly the tail part. This 
explains the better performance of the EVT method than the empirical 
approach. The empirical method is the most flexible setup, bearing no 
risk in model misspecification. However it suffers from the highest 
level of estimation uncertainty among the three methods employed, 
especially when the sample size is limited.

\begin{table}[htbp]
\centering
\begin{tabular}{cccc}
  \hline
& Emp & EVT & SSTD \\ 
  \hline
Emp & 0.750 &  &  \\ 
  EVT & $<0.001$ & 0.907 &  \\ 
  SSTD & $<$0.001 & $<$0.001 & 0.697 \\ 
   \hline
\end{tabular}
	\caption{P-values from backtests of three estimation methods. Diagonal entries correspond to single-method validation backtests. Lower-triangular entries show two-method comparisons, where each p-value tests the null hypothesis that the method in the  column performs better than the method in the row. }
		\label{tab:pvalues}
\end{table}

% \FloatBarrier
% \bibliographystyle{apalike} 
% \bibliography{mybib.bib} 

\clearpage

\appendix

\setcounter{page}{1}
\renewcommand{\thepage}{S\arabic{page}}
\renewcommand{\theequation}{S\arabic{equation}}
\renewcommand{\thelemma}{S\arabic{lemma}}

\section{Proofs}
Throughout the proof, we use the following notations. For two sequences $a_n, b_n$,  $a_n \gtrsim b_n$ means that, $b_n/a_n =O(1)$,  
$a_n \asymp b_n$ means that, $a_n/b_n = O(1)$ and $b_n/a_n = O(1)$ as $n\to\infty$.

\subsection{ Proof of Theorem \ref{Theorem:naive}}

\begin{proof}[Proof of Theorem \ref{Theorem:naive}]
Under $H_0$, we have that
$$
T = \frac{\sum_{i=1}^n Y_i  }{\sqrt{n \widehat{\sigma} } } =\sqrt{ \frac{\sigma }{\widehat{\sigma} }}  \frac{\sum_{i=1}^n (Y_i-\bE Y_i)  }{\sqrt{n\sigma}  },
$$
where 
$$
\sigma =\frac{1}{n}\sum_{i=1}^n \Var(Y_i).
$$

We first show that, as $n\to\infty$,
\begin{equation}\label{s:norm:lyapunov}
\frac{\sum_{i=1}^n (Y_i-\bE Y_i)}{\sqrt{n \sigma } } \stackrel{d}{\to} N(0,1).	
\end{equation}
We intend to apply the Lyapunov central limit theorem to establish the asymptotic normality. To this end, we verify the corresponding conditions. 

By Condition \ref{assum:mixing} and  Corollary 1.1 of \cite{shao1995maximal}, we have that, 
$$
\begin{aligned}
	\bE \suit{Y_i-\bE Y_i}^{4} \le & C \suit{ p \max_{1\le j\le p} \bE \suit{X_{i,j}-\mu_j}^{4} +\suit{p\max_{1\le j\le p} \bE(X_{i,j}-\mu_{j})^2 } ^{2}  }, \\
								= & O(1) \suit{ p \sfmax +p^2 \suit{\ssmax}^2  },
\end{aligned}
$$
uniformly for all $1\le i\le n$. 
By Condition \ref{assum:var:low:bound}, we have that,  as $n\to\infty$,
$$
\bE (Y_i-\bE Y_i) ^2   \ge  c\sum_{j=1}^p \bE(X_{i,j}-\mu_{j})^2 \gtrsim p\ssmin,
$$
uniformly for all $1\le i\le n$. 
By  Condition \ref{assum:lyaponov}, we conclude that, as $n\to\infty$,
\begin{equation}\label{s:lyaponov:condition}
\begin{aligned}
 \frac{\sum_{i=1}^n \bE (Y_i-\bE Y_i)^{4}}{\set{\sum_{i=1}^n\bE(Y_i -\bE Y_i )^2 }^{2}  }  =  O(1) \frac{n p\sfmax  + np^2 \suit{\ssmax}^2   }{ \suit{np\ssmin}^2 } =o(1).
\end{aligned}	
\end{equation}
Thus, the Lyapunov condition holds. By  the Lyapunov central limit theorem, we conclude that \eqref{s:norm:lyapunov} holds.

Next, we show that, as $n\to\infty$,   
\begin{equation}\label{s:var:converge:uni}
	\frac{\widehat{\sigma}}{\sigma } \stackrel{P}{\to} 1.
\end{equation}
Write 
\begin{equation}\label{s:eq:rewrite:sigma00}
\begin{aligned}
\widehat{\sigma} =  	\frac{1}{n}\sum_{i=1}^n \suit{Y_i - \frac{1}{n}\sum_{i=1}^n Y_i  }^2 
 = 	\frac{1}{n}\sum_{i=1}^n \suit{Y_i - \bE Y_i  }^2 - \suit{\frac{1}{n}\sum_{i=1}^n Y_i-\bE Y_i }^2.
 \end{aligned}	
\end{equation}

We first handle the first term on the right-hand side of \eqref{s:eq:rewrite:sigma00}. 
Note that, 
$$
\begin{aligned}
\bE \set{\frac{1}{n \sigma	}\sum_{i=1}^n \suit{Y_i -\bE Y_i  }^2-1} = 0.
\end{aligned}
$$
By \eqref{s:lyaponov:condition}, we have that, as $n\to\infty$, 
$$
\begin{aligned}
	\text{Var} \set{\frac{1}{n\sigma }\sum_{i=1}^n \suit{Y_i - \bE Y_i   }^2-1} = &  \frac{1}{n^2 \sigma^2}   \sum_{i=1}^n \text{Var}(Y_i-\bE Y_i)^2 \\
	\le &  \frac{1}{\sigma^2 n^2} \sum_{i=1}^n  \bE (Y_i-\bE Y_i)^4 \\
	= & O(1) \frac{n p\sfmax  + np^2 \suit{\ssmax}^2   }{ \suit{np\ssmin}^2 } \\
	= & o(1).
\end{aligned}
$$
Combining the limits of the mean and variance,  we obtain that, as $n\to\infty$,
$$
\frac{1}{n \sigma	}\sum_{i=1}^n (Y_i - \bE Y_i  )^2\stackrel{P}{\to} 1.
$$

Next, we handle the second term on the right-hand side of \eqref{s:eq:rewrite:sigma00}. 
By \eqref{s:norm:lyapunov}, we have that, as $n\to\infty$,
$$
\frac{1}{\sigma }\suit{\frac{1}{n}\sum_{i=1}^n Y_i-\bE Y_i}^2 \stackrel{P}{\to} 0.
$$
Thus, we conclude that, \eqref{s:var:converge:uni} holds.

Combining \eqref{s:norm:lyapunov} and \eqref{s:var:converge:uni}, and applying the  Slutsky lemma, we complete the proof.
 \end{proof}

\subsection{Proofs of Theorem \ref{theorem:size:improved} }

Throughout the rest of the proofs, 
let $a_n$ and $b_n$ be two positive sequences satisfying  
\begin{align}
	\frac{b_n^2}{ q \ssmax\log (a_n nd) + \log^4 q\log^2 (a_n nd)}   \to & \infty, \label{s:an:bn:lemma:bound:Y}   \\
	\frac{q \sfmax +q^2 \suit{\ssmax}^2}{a_n}\to &0, \label{s:an:mean:dif}\\ 
	\frac{nq\ssmin }{  b_n^2\log^3 d }  \to &\infty. \label{s:an:upper}
\end{align}
One choice for $a_n$ and $b_n$ is 
\begin{equation}\label{s:def:an:bn}
	\begin{aligned}
a_n =&  \suit{nq \ssmin}^2, \\
b_n =& \suit{\frac{nq\ssmin}{\log^3 d}   \set{ q \ssmax\log (a_n nd) + \log^4 q\log^2 (a_n nd)} }^{1/4}.
\end{aligned}
\end{equation}
By Condition \ref{assum:joint:sequence:lyaponov} and \ref{assum:joint:sequence:highdim}, we can show that, as $n\to\infty$,
 \eqref{s:an:bn:lemma:bound:Y},  \eqref{s:an:mean:dif} and \eqref{s:an:upper} hold. 

Before proving Theorem \ref{theorem:size:improved}, we establish some preliminary lemmas. 

\begin{lemma}\label{lemma:upper:bound:Y}
Assume the same conditions as in Theorem \ref{theorem:size:improved}.   Then,     as $n\to\infty$,
$$
a_n\Pr\suit{ \max_{1\le i\le n} \max_{1\le \ell \le d} \abs{Y_i^{(\ell)}-\bE Y_i^{(\ell)} } \ge b_n }\to 0. 
$$	
\end{lemma}
\begin{proof}[Proof of Lemma \ref{lemma:upper:bound:Y}]
Note that, $\rho$-mixing sequences are also $\alpha$-mixing \citep{bradley2005basic} with 
$$
\alpha_i(m)\le \rho_i(m)/4.
$$ 
By Condition \ref{assum:mixing} and  Theorem 2 of \cite{merlevede2009bernstein}, we have that,   there exist constants $C_1>0, C_2>0, C_3>0$, such that,
	$$
	\begin{aligned}
		& a_n\Pr\suit{ \max_{1\le \ell \le d} \max_{1\le i\le n} \abs{Y_i^{(\ell)}-\bE Y_i^{(\ell)} }\ge b_n }\\ 
		\le & a_n nd\max_{1\le \ell \le d} \max_{1\le i\le n}
		\Pr\suit{\abs{Y_i^{(\ell)}-\bE Y_i^{(\ell)} }\ge b_n } \\
		=& a_nnd\max_{1\le \ell \le d} \max_{1\le i\le n}\Pr\suit{ \abs{\sum_{j\in \mS_{\ell}}  \suit{X_{i,j}-\mu_j  } }\ge  b_n }\\ 
		\le &a_n nd  \max_{1\le \ell \le d} \max_{1\le i\le n} \exp\suit{-C_1 \frac{ b_n^2}{v_{\ell}^2q +C_2+ C_3b_n\log^2q} },
	\end{aligned}
	$$
where
\begin{equation}\label{s:def:v2}
	v_{\ell}^2 = \max_{1\le i\le n}\sup_{ j \in  \mS_{\ell} }\suit{\text{Var}(X_{i,j}) +2 \sum_{ k\in \mS_{\ell},  k> j}\abs{\text{Cov}(X_{i,j},X_{i,k} ) } }.
\end{equation}
By Lemma 2.1 of \cite{shao1995maximal} and Conditions  \ref{assum:mixing},  we have that,  as $n\to\infty$,
\begin{equation}\label{s:eq:bound:v2}
v_{\ell}^2 =  O(1) \ssmax \max_{1\le i\le n} \sum_{m=1}^{\infty} \rho_{i}(m) =O(1)\ssmax,	
\end{equation}
uniformly for all $1\le \ell \le d$. Thus, by \eqref{s:an:bn:lemma:bound:Y}, we have that, 
$$
\frac{ b_n^2}{v_{\ell}^2q +C_2+ C_3b_n\log^2q} \frac{1}{\log (a_n nd)} \to \infty,
$$
uniformly for all $1\le \ell \le  d, 1\le i\le n$. It follows that,
  as $n\to\infty$,
$$
a_n\Pr\suit{ \max_{1\le i\le n} \max_{1\le \ell \le d}\abs{Y_i^{(\ell)}-\bE Y_i^{(\ell)} }\ge  b_n}\to 0. 
$$
\end{proof}

\begin{lemma}\label{lemma:fourth:moment}
Assume the same conditions as in Theorem \ref{theorem:size:improved}.
 Then,  as $n\to\infty$,   
	$$
		\begin{aligned}
		 \bE\suit{Y_i^{(\ell)}-\bE Y_i^{(\ell)} }^4 = O\suit{ q \sfmax+q^2 \suit{\ssmax}^2 },
		\end{aligned}
	$$
uniformly for all $1\le \ell \le d,\ 1\le i\le n$.
\end{lemma}
\begin{proof}[Proof of Lemma \ref{lemma:fourth:moment}]
	By  Condition \ref{assum:mixing} and  Corollary 1.1 of \cite{shao1995maximal}, we have that, there exists a constant $C>0$, such that,  
	$$
\begin{aligned}
	 \bE\suit{Y_i^{(\ell)}-\bE Y_i^{(\ell)} }^4  =& \bE \suit{\sum_{j\in \mS_{\ell} } \suit{X_{i,j}-\mu_j }  }^4  \\
	 \le & C \suit{   q \max_{j\in \mS_{\ell}}  \bE\suit{X_{i,j}-\mu_j}^4+\set{q \max_{ j\in \mS_{\ell} }  \bE\suit{X_{i,j}-\mu_j}^2 }^2  },\\
	 =& O\suit{ q \sfmax+q^2 \suit{\ssmax}^2 }.
\end{aligned}	
	$$ 
	uniformly for all $1\le \ell \le d, 1\le i\le n$. 	
\end{proof}

Denote
$$
\widetilde{\sigma}_{\ell_1\ell_2} = \frac{1}{n}\sum_{i=1}^n \suit{Y_i^{(\ell_1)} -\bE Y_i^{(\ell_1)} }\suit{Y_i^{(\ell_2)} -\bE Y_i^{(\ell_2)}  }, 1\le \ell_1,\ell_2\le d,
$$
as a pseudo-empirical estimator of the quantity
$$
\sigma_{\ell_1\ell_2} = \frac{1}{n}\sum_{i=1}^n  \textnormal{Cov}\suit{Y_i^{(\ell_1)}, Y_i^{(\ell_2)}}.
$$
\begin{lemma}\label{lemma:var:joint}
Assume the same conditions as in Theorem \ref{theorem:size:improved}.  Then, as $n\to\infty$, 
$$
 \max_{1\le \ell_1 \le \ell_2\le d}\abs{\frac{\widetilde{\sigma}_{\ell_1\ell_2} -\sigma_{\ell_1\ell_2}}{ \sqrt{\sigma_{\ell_1\ell_1}\sigma_{\ell_2\ell_2}  } }}=o_P(1/\log^2d).
$$
\end{lemma}

\begin{proof}[Proof of Lemma \ref{lemma:var:joint}]

By Condition \ref{assum:joint:var:low:bound}, we have that,   as $n\to\infty$,
\begin{equation}\label{s:low:bound:var}
\sigma_{\ell\ell}  =\frac{1}{n}\sum_{i=1}^n \Var(Y_i^{(\ell)})  \gtrsim   q\ssmin,
\end{equation}
uniformly for all $1\le \ell \le d$. 
Thus, it suffices to show that, as $n\to\infty$,
\begin{align}
&\max_{1\le \ell_1, \ell_2\le d}\abs{\frac{1}{n}\sum_{i=1}^n \suit{Y_i^{(\ell_1)} -\bE Y_i^{(\ell_1)} }\suit{Y_i^{(\ell_2)} - \bE Y_i^{(\ell_2)} }-\sigma_{\ell_1\ell_2} } =  o_P(1) q\ssmin/\log^2d. \label{s:eq:con:var}
\end{align}
We intend to complete the proof of Lemma \ref{lemma:var:joint}  by showing that, as $n\to\infty$, 
\begin{align}
\max_{1\le \ell_1, \ell_2\le d}\abs{\frac{1}{n}\sum_{i=1}^n Z_i^{(\ell_1)}Z_i^{(\ell_2)}-\frac{1}{n}\sum_{i=1}^n \overline{Z}_i^{(\ell_1)}\overline{Z}_i^{(\ell_2)} } = & o_P(1) q\ssmin/\log^2d \label{s:var:joint1}, \\
\max_{1\le \ell_1, \ell_2\le d}\abs{ \frac{1}{n}\sum_{i=1}^n \overline{Z}_i^{(\ell_1)}\overline{Z}_i^{(\ell_2)} -\bE \suit{ \overline{Z}_i^{(\ell_1)}\overline{Z}_i^{(\ell_2)}}  } = & o_P(1) q\ssmin/\log^2d \label{s:var:joint3}, \\
\max_{1\le \ell_1, \ell_2\le d}\abs{\bE \suit{ Z_i^{(\ell_1)} Z_i^{(\ell_2)}} -\bE \suit{ \overline{Z}_i^{(\ell_1)}\overline{Z}_i^{(\ell_2)}}   } = & o_P(1) q\ssmin/\log^2d \label{s:var:joint2}, 
\end{align}
where 
$$
\begin{aligned}
	Z_i^{(\ell)} =& Y_i^{(\ell)}-\bE Y_i^{(\ell_1)}, \\
	\overline{Z}_i^{(\ell)} = & 	Z_i^{(\ell)}\mI\suit{	Z_i^{(\ell)} \le  b_n },
\end{aligned}
$$
and $b_n$ is defined in \eqref{s:def:an:bn}. 
	
First, we prove 	\eqref{s:var:joint1}.  By Lemma \ref{lemma:upper:bound:Y}, we have that, for any $\varepsilon>0$, 
$$
\begin{aligned}
	&\Pr\suit{\max_{1\le \ell_1, \ell_2\le d}\abs{\frac{1}{n}\sum_{i=1}^n Z_i^{(\ell_1)}Z_i^{(\ell_2)}-\frac{1}{n}\sum_{i=1}^n \overline{Z}_i^{(\ell_1)}\overline{Z}_i^{(\ell_2)} }> \varepsilon  qs_{2,n}/\log^2 d  } \\
\le &	\Pr\suit{ \max_{1\le i\le n} \max_{1\le \ell \le d}\abs{Y_i^{(\ell)}-\bE Y_i^{(\ell)} }\ge b_n }\to 0. 
\end{aligned}
$$
Thus, \eqref{s:var:joint1} holds. 

Next, we prove \eqref{s:var:joint3}.
By the Bernstein inequality (page 855 of \cite{shorack1986empirical}),
 we have that, for any $\varepsilon>0$,
$$
\begin{aligned}
& \Pr\suit{\max_{1\le \ell_1, \ell_2\le d}\abs{ \frac{1}{n}\sum_{i=1}^n \overline{Z}_i^{(\ell_1)}\overline{Z}_i^{(\ell_2)} -\bE \suit{ \overline{Z}_i^{(\ell_1)}\overline{Z}_i^{(\ell_2)}} }   \ge \varepsilon q  \ssmin /\log^2d}	\\
\le & d^2 \max_{1\le \ell_1, \ell_2\le d} \Pr\suit{ \abs{ \frac{1}{n}\sum_{i=1}^n \overline{Z}_i^{(\ell_1)}\overline{Z}_i^{(\ell_2)} -\bE \suit{ \overline{Z}_i^{(\ell_1)}\overline{Z}_i^{(\ell_2)}} }   \ge \varepsilon  q \ssmin/\log^2d} \\
\le & d^2 \max_{1\le \ell_1, \ell_2\le d} \exp\suit{-C_1\varepsilon^2 \frac{nq^2\suit{\ssmin}^2 /\log^4 d }{C_2\frac{1}{n}\sum_{i=1}^n s_{i,\ell_1\ell_2}^2 +C_3 \varepsilon  b_n^2 q \ssmin /\log^2 d   }   },
\end{aligned}
$$
where 
$$
s_{i,\ell_1\ell_2}^2= \bE \suit{  \overline{Z}_i^{(\ell_1)}\overline{Z}_i^{(\ell_2)} }^2.
$$
By using the  Cauchy-Schwarz inequality and Lemma \ref{lemma:fourth:moment}, we have that, as $n\to\infty$, 
$$
s_{i,\ell_1\ell_2}^2 \le  \bE \suit{  Z_i^{(\ell_1)}Z_i^{(\ell_2)} }^2 \le  \sqrt{  \bE \suit{ Z_i^{(\ell_1)}}^4 \bE \suit{Z_i^{(\ell_2)} }^4} =O(1)\suit{q\sfmax +q^2\suit{\ssmax}^2 },
$$
uniformly for all $1\le \ell_1\le \ell_2\le d, 1\le i\le n$.  By Condition \ref{assum:joint:sequence:lyaponov}  and \eqref{s:an:upper}, 
  we have that, as $n\to\infty$,
$$
\frac{nq^2\suit{\ssmin}^2 }{\log^4 d  \suit{q\sfmax +q^2\suit{\ssmax}^2 }}\frac{1}{\log d}\to \infty, \quad \frac{nq\ssmin/\log^2 d }{  b_n^2 } \frac{1}{\log d} \to\infty. 
$$
It follows that,  as $n\to\infty$,
$$
 \Pr\suit{\max_{1\le \ell_1, \ell_2\le d}\abs{\frac{1}{n}\sum_{i=1}^n (Y_i^{(\ell_1)}- \bE Y_i^{(\ell_1)}  )(Y_i^{(\ell_2)}-\bE Y_i^{(\ell_2)})-\sigma_{\ell_1\ell_2} } \ge \varepsilon  q\ssmin/\log^2d}\to 0.	
 $$
 Then, \eqref{s:var:joint3} holds.

Finally, we prove \eqref{s:var:joint2}. Note that, 
\begin{align*}
	&\abs{\bE \suit{ Z_i^{(\ell_1)} Z_i^{(\ell_2)}} -\bE \suit{ \overline{Z}_i^{(\ell_1)}\overline{Z}_i^{(\ell_2)}}   } \\
	\le &  \abs{\bE \set{ Z_i^{(\ell_1)} Z_i^{(\ell_2)}\mI\suit{Z_i^{(\ell_1)} > b_n }} } 
	+\abs{\bE \set{ Z_i^{(\ell_1)} Z_i^{(\ell_2)}\mI\suit{Z_i^{(\ell_2)} >  b_n } }}.
\end{align*}
By the generalized H\"older inequality, Lemma \ref{lemma:upper:bound:Y}, Lemma \ref{lemma:fourth:moment},  and   \eqref{s:an:mean:dif}, we have that, as $n\to\infty$, 
$$ 
\begin{aligned}
	&\abs{\bE \set{ Z_i^{(\ell_1)} Z_i^{(\ell_2)}\mI\suit{Z_i^{(\ell_1)} >  b_n} } } \\
	\le & \set{\bE \suit{ Z_i^{(\ell_1)}}^4}^{1/4} \set{\bE \suit{ Z_i^{(\ell_2)}}^4}^{1/4} \set{\Pr\suit{{Z_i^{(\ell_1)} >   b_n } }  }^{1/2}  \\
	= &o(1)  \suit{q \sfmax +q^2 \suit{\ssmax}^2 }^{1/2} \suit{\frac{1}{ a_n}}^{1/2}\\
		=& o(1), 
\end{aligned}
$$
uniformly for all $1\le \ell_1\le \ell_2\le d$. 
Similarly, we have that, 
as $n\to\infty$,
$$
\abs{\bE \set{ Z_i^{(\ell_1)} Z_i^{(\ell_2)}\mI\suit{Z_i^{(\ell_2)} >  b_n}} } = o(1),
$$
uniformly for all $1\le \ell_1\le \ell_2\le d$. 
Thus, \eqref{s:var:joint2} holds.

 The proof of Lemma \ref{lemma:var:joint} is then complete. 
\end{proof}

\begin{lemma}\label{lemma:convergence:var}
Assume the same conditions as in Theorem \ref{theorem:size:improved}. Then, as $n\to\infty$, 
	$$
	\max_{1\le \ell\le d}\abs{\frac{\widehat{\sigma}_{\ell\ell}}{\sigma_{\ell\ell}}-1}= o_P(1/\log^2 d).
	$$
\end{lemma}
\begin{proof}[Proof of Lemma \ref{lemma:convergence:var}]
Write 
$$
\begin{aligned}
\widehat{\sigma}_{\ell\ell} -\sigma_{\ell\ell}=&   \widehat{\sigma}_{\ell\ell}- \widetilde{\sigma}_{\ell\ell}+ \widetilde{\sigma}_{\ell\ell}-\sigma_{\ell\ell} \\
=&  	\frac{1}{n}\sum_{i=1}^n (Y_i^{(\ell)}-\frac{1}{n}\sum_{i=1}^n Y_i^{(\ell)})^2 - \frac{1}{n}\sum_{i=1}^n \suit{Y_i^{(\ell)} -\bE Y_i^{(\ell)} }^2+ \widetilde{\sigma}_{\ell\ell}-\sigma_{\ell\ell} \\
=&\suit{\frac{1}{n}\sum_{i=1}^n Y_i^{(\ell)}  - \bE Y_i^{(\ell)}}^2+\widetilde{\sigma}_{\ell\ell}-\sigma_{\ell\ell} \\
=&  \suit{\frac{1}{n}\sum_{i=1}^n  Z_i^{(\ell)}}^2+\widetilde{\sigma}_{\ell\ell}-\sigma_{\ell\ell}.
\end{aligned}
$$
	By Lemma \ref{lemma:var:joint} and \eqref{s:low:bound:var}, it suffices to show that, as $n\to\infty$, 
\begin{equation}\label{s:con:var:1}
	\suit{\frac{1}{n}\sum_{i=1}^n  Z_i^{(\ell)}}^2 =  o_P(1) q \ssmin/\log^2d. 
\end{equation}
We intend to prove a stronger result than  \eqref{s:con:var:1}:
\begin{equation}\label{s:con:var:2}
	\max_{1\le \ell \le d} \abs{ \frac{1}{n}\sum_{i=1}^n  Z_i^{(\ell)}} =o_P(1) \sqrt{\frac{q\ssmin}{n}}\log^{3/2} d.
\end{equation}
	By Condition \ref{assum:mixing} and  Theorem 2 of \cite{merlevede2009bernstein}, we have that,   there exist  constants $C_1>0, C_2>0, C_3>0$, such that,
	for any $x>0$,
	$$
	\begin{aligned}
	 \Pr\suit{   \max_{1\le \ell \le d}\abs{\frac{1}{n}\sum_{i=1}^n  Z_i^{(\ell)}} \ge x } 
		\le & d \max_{1\le \ell \le d} \Pr\suit{\abs{\sum_{i=1}^n  Z_i^{(\ell)}} \ge nx } \\
		=& d\max_{1\le \ell \le d} \Pr\suit{ \abs{\sum_{i=1}^n \sum_{j\in \mS_{\ell} }  \suit{X_{i,j}- \mu_{j}  } }\ge nx}\\ 
		\le & d\max_{1\le \ell \le d} \exp\suit{-C_1 \frac{ n^2x^2}{nq v_{\ell}^2 +C_2+ C_3nx\log^2(nq)} },
	\end{aligned}
	$$
where 
$v_{\ell}^2 $ is defined in \eqref{s:def:v2}. 
Take 
$$
x = 	\sqrt{\frac{q\ssmin }{n}}\log^{3/2} d.
$$ 
By \eqref{s:eq:bound:v2},   Conditions \ref{assum:joint:sequence:lyaponov} and \ref{assum:joint:sequence:highdim}, we have that, 
 as $n\to\infty$,
$$
\begin{aligned}
	&\frac{n^2x^2}{nqv_{\ell}^2 }\frac{1}{\log d}  =\frac{\ssmin}{\ssmax}\log^2d  \to \infty,  \\ 
	 &\frac{n^2x^2}{\log d} = nq\ssmin \log^2d   \to \infty,  \\
&\frac{n^2x^2}{nx\log^2 (nq)} \frac{1}{\log d}  = \sqrt{nq\ssmin} \frac{\log^{1/2} d}{\log^2(nq)}   \to\infty.
\end{aligned}
$$
Thus, we have that,  as $n\to\infty$,
$$
\Pr\suit{   \max_{1\le \ell \le d}\abs{\frac{1}{n}\sum_{i=1}^n  Z_i^{(\ell)}} \ge \sqrt{\frac{q\ssmin }{n}}\log^{3/2} d }\to 0. 
$$
It follows that, \eqref{s:con:var:2} holds. 
The proof is then complete.
\end{proof}

\begin{proof}[Proof of Theorem \ref{theorem:size:improved}]
Under $H_0$, we have that 
$$
T^{(\ell)} = \sqrt{\frac{\sigma_{\ell\ell}}{\widehat{\sigma}_{\ell\ell}} }  \frac{1}{\sqrt{n}}\sum_{i=1}^n \frac{ Y_i^{(\ell)}}{\sqrt{\sigma_{\ell\ell}}} = \sqrt{\frac{\sigma_{\ell\ell}}{\widehat{\sigma}_{\ell\ell}} }  \frac{1}{\sqrt{n}}\sum_{i=1}^n \frac{ Z_i^{(\ell)}}{\sqrt{\sigma_{\ell\ell}}}.
$$
By  \eqref{s:low:bound:var}, \eqref{s:con:var:2}   and Lemma \ref{lemma:convergence:var},    we have that, as $n\to\infty$, 
\begin{equation}\label{s:eq:to:iid}
\begin{aligned}
		T^{(\ell)} =& \frac{1}{\sqrt{n}}\sum_{i=1}^n \frac{ Z_i^{(\ell)}}{\sqrt{\sigma_{\ell\ell}}} + o(1)\frac{1}{\log^2d } \frac{\sqrt{n}}{\sqrt{\sigma_{\ell\ell}}} \frac{1}{n}\sum_{i=1}^n Z_i^{(\ell)} \\
	    =& \frac{1}{\sqrt{n}}\sum_{i=1}^n \frac{ Z_i^{(\ell)}}{\sqrt{\sigma_{\ell\ell}}} + o(1)\frac{1}{\log^2d } \frac{\sqrt{n}}{\sqrt{q\ssmin}} \sqrt{\frac{q \ssmin }{n}}\log^{3/2} d \\
	    =&\frac{1}{\sqrt{n}}\sum_{i=1}^n \frac{ Z_i^{(\ell)}}{\sqrt{\sigma_{\ell\ell}}}  +o_P(1)\log^{-1/2} d,
\end{aligned}  
\end{equation}
uniformly for all $1\le \ell \le d$.

We intend to apply High dimensional central limit theorem   to approximate the distribution of 
$$
 \suit{\frac{1}{\sqrt{n}}\sum_{i=1}^n \frac{Z_i^{(1)} }{\sqrt{\sigma_{11}}}, \dots,  \frac{1}{\sqrt{n}}\sum_{i=1}^n \frac{ Z_i^{(d)} }{\sqrt{\sigma_{dd}}}}.
$$
For this purpose, we verify the conditions in Theorem 1 in \cite{chernozhukov2023high}. 
Let $B_n$ be a real sequence satisfying 
\begin{align}
	\frac{1}{B_n^2}\set{ \suit{\frac{\ssmax}{\ssmin}}^2 +\frac{\sfmax }{q\suit{\ssmin}^2 }+ \frac{\log^4 q}{q\ssmin } }   \to &0  ,   \label{s:Bn:lower} \\
	 B_n^2 \frac{ \log^5(dn)}{n}\to &0. \label{s:Bn:upper}
\end{align}
as $n\to\infty$. By Condition \ref{assum:joint:sequence:lyaponov} and Condition \ref{assum:joint:sequence:highdim}, we have that, 
$$
\set{ \suit{\frac{\ssmax}{\ssmin}}^2 +\frac{\sfmax }{q\suit{\ssmin}^2 }+ \frac{\log^4 q}{q\ssmin } }    \frac{ \log^5(dn)}{n} \to 0. 
$$
Therefore a sequence $B_n$ satisfying \eqref{s:Bn:lower} and \eqref{s:Bn:upper} exists. 
One convenient choice is    
$$
B_n^2 = \sqrt{  \set{ \suit{\frac{\ssmax}{\ssmin}}^2 +\frac{\sfmax }{q\suit{\ssmin}^2 }+ \frac{\log^4 q}{q\ssmin }   }/ \suit{ \frac{ \log^5(dn)}{n}   } }.
$$

 By \eqref{s:low:bound:var},  Lemma \ref{lemma:fourth:moment} and Condition \eqref{s:Bn:lower}, we have that, as $n\to\infty$,
$$
\begin{aligned}
		\bE \suit{  \frac{Z_i^{(\ell)}}{\sqrt{\sigma_{\ell\ell}}}}^2 =&1, \\
		\max_{1\le \ell \le d}	\bE \suit{\frac{ Z_i^{(\ell)}}{\sqrt{\sigma_{\ell\ell}}}}^4  =&O(1)\set{\suit{\frac{\ssmax}{\ssmin}}^2 +\frac{\sfmax }{q \suit{\ssmin}^2 } } =o(1) B_n^2. 
\end{aligned}
$$

Define 
$$
t  = 2\frac{1}{\sqrt{\sigma_{\ell\ell}} B_n}.
$$
By \eqref{s:low:bound:var} and \eqref{s:Bn:lower}, we have that, as $n\to\infty$,
$$
t  = o(1) \frac{1}{\log^2 q}.  
$$

 By Condition \ref{assum:mixing} and Theorem 2 of \cite{merlevede2009bernstein}, we have that, for some $C_1>0, C_2>0$,
$$
\begin{aligned}
\log \bE \exp  \suit{  2\frac{Z_i^{(\ell)}}{\sqrt{\sigma_{\ell\ell}}B_n }} \le \frac{ C_1t^2(qv_{\ell}^2+1)}{1-C_2t\log^2 q}, 
\end{aligned}
$$
where $v_{\ell}^2$ is defined in \eqref{s:def:v2}.  By \eqref{s:eq:bound:v2} and \eqref{s:Bn:lower}, we have that, as $n\to\infty$, 
$$
\frac{ C_1t^2(qv_{\ell}^2+1)}{1-C_2t\log^2 q} = o(1). 
$$ 
Thus, we have that, for sufficiently large $n$,
$$
 \bE \exp  \suit{  2\frac{ Z_i^{(\ell)}}{\sqrt{\sigma_{\ell\ell}}B_n }}  \le \sqrt{2}.
$$
Similarly, we have that, for sufficiently large $n$,
$$
  \bE \exp  \suit{-  2\frac{Z_i^{(\ell)}}{\sqrt{\sigma_{\ell\ell}}B_n }}  \le \sqrt{2}.
$$
By using the  Cauchy-Schwarz inequality, we have that, for sufficiently large $n$, 
$$
\begin{aligned}
  \bE \exp  \suit{  \frac{\abs{  Z_i^{(\ell)} }}{\sqrt{\sigma_{\ell\ell}}B_n }} =&  \bE \exp  \suit{  \frac{ Z_i^{(\ell)} }{\sqrt{\sigma_{\ell\ell}}B_n }} I\suit{Z_i^{(\ell)}>0}  +  \bE \exp  \suit{  \frac{ -Z_i^{(\ell)} }{\sqrt{\sigma_{\ell\ell}}B_n }} I\suit{Z_i^{(\ell)}\le 0}\\
  \le & \sqrt{2} \suit{ \sqrt{\Pr\suit{Z_i^{(\ell)}>0}} +\sqrt{\Pr\suit{Z_i^{(\ell)}\le 0}} } \\
  \le & \sqrt{2} \sqrt{2} =2.
\end{aligned}
$$

By Theorem 1 of \cite{chernozhukov2023high} and Condition \eqref{s:Bn:upper}, we have that, 
$$
\begin{aligned}
	&\abs{\Pr\set{ \suit{\frac{1}{\sqrt{n}}\sum_{i=1}^n \frac{Z_i^{(1)} }{\sqrt{\sigma_{11}}}, \dots,  \frac{1}{\sqrt{n}}\sum_{i=1}^n \frac{ Z_i^{(d)} }{\sqrt{\sigma_{dd}}}}\in A} - \Pr\set{N(0,\bolSigma)\in A }}  \\
	=& O(1) \suit{\frac{B_n^2\log^5(dn)}{n }}^{1/4} =o(1), 
\end{aligned}
$$
uniformly for all $A\in \mathcal{A}$, where $\mathcal{A}$ is the collection of closed rectangles in $\mathbb{R}^p$,
$$
\mathcal{A} = \set{\prod_{j=1}^d [a_j,b_j]:-\infty\le a_j\le b_j\le \infty, j=1,\dots,d},
$$
and $\bolSigma$ is a $d\times d$ matrix with elements
$$
\bolSigma_{\ell_1,\ell_2} = \frac{1}{n}\sum_{i=1}^n  \frac{Z_i^{(\ell_1} Z_i^{\ell_2}}{\sqrt{\sigma_{\ell_1\ell_1} \sigma_{\ell_2\ell_2} }}.
$$

By Lemma 1 of \cite{chernozhukov2023high} and \eqref{s:eq:to:iid}, we have that,
\begin{equation}\label{s:eq:norm:T}
	\abs{\Pr\set{ \suit{T^{(1)}, \dots, T^{(\ell)} }\in A} - \Pr\set{N(0,\bolSigma)\in A }}  
 =o(1). 
\end{equation}

Next, we handle the bootstrapped versions.
Given the  data $(X_{i,1},\dots, X_{i,p}), i=1,\dots,n$, the bootstrapped versions have the following conditional distribution:
$$
\suit{T_B^{(1)}, \dots, T_B^{(d)}}\sim N(0,\widehat{\bolSigma}),
$$
where 
$$
\widehat{\bolSigma}_{\ell_1\ell_2} =\frac{1}{\sqrt{\widehat{\sigma}_{\ell_1\ell_1}\widehat{\sigma}_{\ell_2\ell_2}}}  \frac{1}{n}\sum_{i=1}^n Y_i^{(\ell_1)}Y_i^{(\ell_2)}, \quad 1\le \ell_1\le \ell_2 \le d.
$$
By Lemma \ref{lemma:var:joint} and Lemma \ref{lemma:convergence:var}, we have that, under $H_0$,  as $n\to\infty$, 
$$
\abs{\widehat{\bolSigma}_{\ell_1\ell_2}- \bolSigma_{\ell_1\ell_2}} =o_P(1/\log^2 d),
$$
uniformly for all $1\le \ell_1,\ell_2\le d$. 
Thus, by Proposition 2.1 of \cite{chernozhukov2022improved}, we have that,  as $n\to\infty$,
$$
\Pr\suit{N(0,\bolSigma)\in A } - \Pr\suit{N(0,\widehat{\bolSigma})\in A } =o(1).
$$
Combining with \eqref{s:eq:norm:T}, we have that, as $n\to\infty$,
 $$
	\abs{\Pr\suit{ \suit{T^{(1)}, \dots, T^{(\ell)} }\in A} - \Pr\suit{N(0,\widehat{\bolSigma})\in A }}  
 =o(1). 
$$
The proof is then complete.
 \end{proof}

 \subsection{Proof of Theorem \ref{theorem:power:improved}}
 
 \begin{proof}[Proof of Theorem \ref{theorem:power:improved}]
 	Recall that, given the  data $(X_{i,1},\dots, X_{i,p}), i=1,\dots,n$,  
$$
\suit{T_B^{(1)}, \dots, T_B^{(d)}}\sim N(0,\widehat{\bolSigma}),
$$
where 
$$
\widehat{\bolSigma}_{\ell_1\ell_2} =\frac{1}{\sqrt{\widehat{\sigma}_{\ell_1\ell_1}\widehat{\sigma}_{\ell_2\ell_2}}}  \frac{1}{n}\sum_{i=1}^n  Y_i^{(\ell_1)} Y_i^{(\ell_2)}   , \quad 1\le \ell_1\le \ell_2 \le d.
$$

By using the standard result on Gaussian maximum (see e.g. (B.31) of \cite{chang2017simulation}), we have that, for sufficiently large $p$,
\begin{equation}\label{s:upper:bound:c:alpha}
c_{\alpha} \le \msuit{\set{(1+(2\log d)^{-1}} \sqrt{2\log d} +\sqrt{2\log (1/\alpha)}} \max_{1\le \ell\le d} \suit{\widehat{\bolSigma}_{\ell\ell}}^{1/2}.
\end{equation}
 	
Write
$$
\begin{aligned}
	\widehat{\bolSigma}_{\ell\ell} = & \frac{1}{\widehat{\sigma}_{\ell\ell} } \frac{1}{n}\sum_{i=1}^n \suit{Y_i^{(\ell)}}^2   \\
	 =&  \frac{1}{\widehat{\sigma}_{\ell\ell} } \set{\frac{1}{n}\sum_{i=1}^n \suit{Y_i^{(\ell)} -\bE Y_i^{(\ell)} }^2 +\suit{\bE Y_i^{(\ell)}}^2 - 2\bE Y_i^{(\ell)}  \frac{1}{n}\sum_{i=1}^n \suit{Y_i^{(\ell)} -\bE Y_i^{(\ell)} }  }\\
	 =&\frac{\widetilde{\sigma}_{\ell\ell}}{\widehat{\sigma}_{\ell\ell}} -2 \frac{\delta_{\ell}}{\widehat{\sigma}_{\ell\ell}} \frac{1}{n}\sum_{i=1}^n Z_i^{(\ell)}   +\frac{\delta_{\ell}^2}{\widehat{\sigma}_{\ell\ell}},
\end{aligned}
$$
where 
$
\delta_{\ell} = \bE Y_i^{(\ell)} =\sum_{j\in \mS_{\ell}}\mu_j.
$

Denote $\delta = \max_{1\le \ell \le d} \abs{\delta_{\ell}}.$
By Lemma \ref{lemma:convergence:var}, \eqref{s:low:bound:var}, \eqref{s:con:var:2}, we have that, for any $\varepsilon>0$ and some $C>0$,
 with probability tending to $1$, 
$$
\max_{1\le \ell \le d} \widehat{\bolSigma}_{\ell\ell} \le \suit{1+\varepsilon}\suit{1+ \frac{ \varepsilon \delta }{ \sqrt{q\ssmin} } +\frac{C\delta^2}{q \ssmin }},
$$
and hence for sufficiently large $p$, 
$$
c_{\alpha}\le \sqrt{2(1+\varepsilon)\log d}\suit{1+ \frac{ \varepsilon \delta }{\sqrt{q\ssmin} } +\frac{C\delta^2}{q\ssmin}  }^{1/2}.
$$

Denote 
$$
\ell^* = \arg\max_{\ell} \abs{\delta_{\ell}}.
$$
Without loss of generality, assume that, $\delta_{\ell^*}>0$.
Write 
$$
\begin{aligned}
	T^{(\ell^*)} = & \frac{\sum_{i=1}^n Y_i^{(\ell^*)}    }{\sqrt{n \widehat{\sigma}_{\ell^*\ell^*} }}  \\ 
			 =& 	\frac{\sum_{i=1}^n (Y_i^{(\ell^*)} - \bE Y_i^{(\ell^*)}   ) }{\sqrt{n \widehat{\sigma}_{\ell^*\ell^*} }} +\frac{\sqrt{n}}{\sqrt{\widehat{\sigma}_{\ell^*\ell^*}}}\delta.
\end{aligned}
$$
 It follows that,
$$
\begin{aligned}
\Pr\suit{\max_{1\le \ell \le d} \abs{ T^{(\ell)} }> c_{\alpha} } 
\ge & \Pr\suit{ \abs{ T^{(\ell^*)} }> c_{\alpha} }  \\
=& 1- \Pr\suit{ \abs{ T^{(\ell^*)} }\le  c_{\alpha} } \\
\ge & 1- \Pr\suit{   \frac{\sum_{i=1}^n (Y_i^{(\ell^*)}-\bE Y_i^{(\ell^*)}  ) }{\sqrt{n \widehat{\sigma}_{\ell^*\ell^*} }} \le  -\suit{  \frac{\sqrt{n}}{\sqrt{\widehat{\sigma}_{\ell^*\ell^*}}}\delta -c_{\alpha}} } \\
\end{aligned}
$$

By Corollary 1.1 of \cite{shao1995maximal}, we have that, for some $K>0$, 
$$
\sigma_{\ell\ell} = \frac{1}{n}\sum_{i=1}^n \text{Var}(Y_i^{(\ell)}) \le K q\ssmax. 
$$
Thus, for sufficiently large $p$,  
with probability tending to 1, we have that,
$$
\begin{aligned}
	&\frac{\sqrt{n}}{\sqrt{\widehat{\sigma}_{\ell^*\ell^*}}}\delta - c_{\alpha} \\
	\ge &\frac{\sqrt{n}}{\sqrt{K q\ssmax  }}(1-\varepsilon) \delta - \sqrt{2(1+\varepsilon)\log d}\sqrt{ 1+ \frac{ \varepsilon \delta }{ \sqrt{q\ssmin} } +\frac{C\delta^2}{q\ssmin}}\\
	\ge &   \frac{\sqrt{n}}{\sqrt{q\ssmax}}\delta \set{ \frac{1-\varepsilon}{\sqrt{K}} - \sqrt{2(1+\varepsilon)\log d} \suit{\frac{q\ssmax }{n\delta^2}+\frac{\varepsilon \sqrt{q} \ssmax }{n\delta \sqrt{\ssmin}} +\frac{C\ssmax}{n\ssmin}}^{1/2}  }\\
	\ge  & \frac{\sqrt{n}}{\sqrt{q\ssmax}}\delta \set{  \frac{1-\varepsilon}{\sqrt{K}} -\sqrt{2(1+\varepsilon)\log d} \suit{\frac{1}{ \lambda\log d}+\frac{\varepsilon \sqrt{\ssmax}  }{\sqrt{n\lambda \log d\ssmin } } +\frac{C\ssmax }{n\ssmin} }^{1/2} } 
\end{aligned}
$$
By Conditions \ref{assum:joint:sequence:lyaponov} and  \ref{assum:joint:sequence:highdim}, and the assumptions that $\lambda \to \infty$,  we have that, as $n\to\infty$,
$$
\sqrt{2(1+\varepsilon)\log d} \suit{\frac{1}{ \lambda\log d}+\frac{\varepsilon \sqrt{\ssmax}  }{\sqrt{n\lambda \log d\ssmin } } +\frac{C\ssmax }{n\ssmin} }^{1/2} \to 0,
$$
and 
$$
\frac{\sqrt{n}}{\sqrt{q\ssmax}}\delta \to \infty.
$$
By taking $\varepsilon<1$, we have that, under $H_1$, with probability tending to 1, 
$$
\frac{\sqrt{n}}{\sqrt{\widehat{\sigma}_{\ell^*\ell^*}}}\delta - c_{\alpha} \to \infty.
$$

Moreover, note that, as $n\to\infty$,
$$
 \frac{\sum_{i=1}^n (Y_i^{(\ell^*)} - \bE Y_i^{(\ell^*)}    ) }{\sqrt{n \widehat{\sigma}_{\ell^*\ell^*} }} \to N(0,1).
$$
Thus, we have that, as $n\to\infty$,
$$
\Pr\suit{   \frac{\sum_{i=1}^n (Y_i^{(\ell^*)}-\bE Y_i^{(\ell^*)}  ) }{\sqrt{n \widehat{\sigma}_{\ell^*\ell^*} }} \le   -\suit{  \frac{\sqrt{n}}{\sqrt{\widehat{\sigma}_{\ell^*\ell^*}}}\delta -c_{\alpha}}  }\to 0,
$$
and hence
$$
\Pr\suit{\max_{1\le \ell \le d} \abs{ T^{(\ell)} }> c_{\alpha} }\to 1.
$$
 \end{proof}

\subsection{Proofs of Other results}

\begin{proof}[Proof of Lemma \ref{lemma:main:subsets}]
	We only need to prove sufficiency as the direction  of necessity is trivial. 
Define a mapping $T: \set{1,\dots,p} \to \set{1,\dots,p}$ by 
$$
T(j) = (j-1 \bmod p)+1.
$$	

Then for each $j\in \set{1,\dots, p}$, we can write
\begin{equation*}
	S_j := \mu_{j}+\mu_{T(j+1)}+\cdots+\mu_{T(j+q-1)}.
\end{equation*}
Under the null that $S_j=0$ for all $j=1,\dots,p$, we have that
\begin{equation*}
 0 = S_{j+1}-S_j = \mu_{T(j+q)}-\mu_j, \ \text{for all} \ j=1,\dots, p, 	
\end{equation*}
which implies
\begin{equation}\label{s:subsets:S}
	 \mu_{T(j+q)} = \mu_j,  \ \text{for all} \ j=1,\dots,p.
\end{equation}

Define a mapping  $T_q: \set{1,\dots,p} \to \set{1,\dots,p}$ by 
$$
T_q(j):= T(j+q) =  (j+q-1 \bmod p)+1.
$$
We get that $\mu_{\cdot}$ is invariant under the mapping $T_q$, i.e. $\mu_{T_q(j)} = \mu_j$ for all $j=1,\dots,p$.

Define $T_q^k(j) = T_q(T_q^{k-1}(j))$ for any integer $k\ge 2$. 
Consider the orbit of any index $j_0$ under repeated application of $T_q$, 
$$
j_0, \ T_q(j_0), \ T_q^2(j_0), \ \dots, \  T_q^k(j_0),\  \dots.  
$$

As the number of index is limited, the orbit must return to a certain value $i_0$ after a finite number of iterations. That is, there exists a positive integer $k$ such that 
$$
T_q^k(i_0) = i_0.
$$
This is equivalent to 
$$
i_0+kq \equiv i_0 \ (\bmod\  p)  \quad  \Longleftrightarrow  \quad  kq \equiv 0 \ (\bmod \ p). 
$$

Since 
$gcd(p,q)=1$, the least positive   $k$ satisfying 
$kq\equiv 0 \ ( \bmod \ p)$ is $k=p$. 
Hence the orbit must contain all $p$  
 indices before returning to  $i_0$.  In other words, iterating $T_q$ produces a single cycle of length 
$p$. Recall that $\mu_{\cdot}$ is invariant under $T_q$, i.e.  the relation \eqref{s:subsets:S}. We thus conclude that
  $$
  \mu_1=\cdots = \mu_p = 0. 
  $$
\end{proof}

\FloatBarrier
\bibliographystyle{apalike} 
\bibliography{mybib}

\end{document}